\documentclass[superscriptaddress,twocolumn]{revtex4-2}

\usepackage[dvipsnames,table]{xcolor}
\usepackage{hyperref}
\hypersetup{
  breaklinks=true,
  colorlinks=true,
  allcolors=BlueViolet,
}

\usepackage{graphicx}
\graphicspath{{./figures/}} %
\usepackage[caption=false]{subfig} %
\newcommand{\sref}[1]{\protect\subref{#1}} %

\usepackage{multirow} %
\usepackage{rotating} %

\usepackage[inline]{enumitem} %
\setlist[enumerate]{label=(\arabic*)} %
\setlist{nosep} %

\usepackage{braket,bm,amssymb}
\usepackage{stmaryrd}  %
\usepackage{amsmath}
\DeclareMathOperator{\rows}{rows}
\DeclareMathOperator{\integer}{\texttt{int}}

\DeclareMathOperator{\poly}{poly}
\DeclareMathOperator{\polylog}{polylog}
\DeclareMathOperator{\diag}{diag}

\renewcommand{\sp}[1]{\left[#1\right]}
\DeclareMathOperator{\Tr}{Tr}
\newcommand{\abs}[1]{\left\lvert#1\right\rvert}
\newcommand{\norm}[1]{\left\lVert#1\right\rVert}
\makeatletter
\newcommand{\op}[1]{\@ifnextchar\bgroup{\op@ii{#1}}{\op@i{#1}}}
\newcommand{\op@i}[1]{\mathinner{|{#1}\rangle\!\langle{#1}|}}
\newcommand{\op@ii}[2]{\mathinner{|{#1}\rangle\!\langle{#2}|}}
\makeatother

\usepackage{dsfont}

\usepackage[export]{adjustbox}

\usepackage[capitalize]{cleveref}

\usepackage[yyyymmdd]{datetime} %

\usepackage{amsthm}
\newtheorem{lemma}{Lemma}

\newcommand{\GTAR}{Global Technology Applied Research, JPMorganChase, New York, NY 10017, USA}
\newcommand{\QUK}{Quantinuum, Partnership House, Carlisle Place, London SW1P 1BX, UK}
\newcommand{\QCO}{Quantinuum, 303 S. Technology Ct., Broomfield, Colorado 80021, USA}

\begin{document}

\title{Fault-tolerant execution of error-corrected quantum algorithms}

\author{Michael A. Perlin}
\email{michael.perlin@jpmchase.com}
\affiliation{\GTAR}
\author{Zichang He}
\email{zichang.he@jpmchase.com}
\affiliation{\GTAR}
\author{Anthony Alexiades Armenakas}
\affiliation{\GTAR}
\author{Pablo Andres-Martinez}
\affiliation{\QUK}
\author{Tianyi Hao}
\affiliation{\GTAR}
\author{Dylan A. Herman}
\affiliation{\GTAR}
\author{Yuwei Jin}
\affiliation{\GTAR}
\author{Karl Mayer}
\affiliation{\QCO}
\author{Chris N. Self}
\affiliation{\QUK}
\author{David Amaro}
\affiliation{\QUK}
\author{Ciaran Ryan-Anderson}
\affiliation{\QCO}
\author{Ruslan Shaydulin}
\email{ruslan.shaydulin@jpmchase.com}
\affiliation{\GTAR}

\begin{abstract}
  Scaling up quantum algorithms to tackle high-impact problems in science and industry requires quantum error correction and fault tolerance.
  While progress has been made in experimentally realizing error-corrected primitives, the end-to-end execution of logical quantum algorithms using only fault-tolerant (FT) components has remained out of reach.
  We demonstrate the FT and error-corrected execution of two quantum algorithms, the Quantum Approximate Optimization Algorithm (QAOA) and the Harrow-Hassidim-Lloyd (HHL) algorithm applied to the Poisson equation, on Quantinuum H2 and Helios trapped-ion quantum processors using the $\llbracket 7,1,3 \rrbracket$ Steane code.
  For QAOA circuits on 5 and 6 logical qubits, we show performance improvements from increasing the number of QAOA layers and the number of $T$ gates used to approximate logical rotations, despite increased physical circuit complexity.
  We further show that QAOA circuits with up to 8 logical qubits and 9 logical $T$ gates perform similarly to unencoded circuits.
  For the largest QAOA circuits we run, with 12 logical (97 physical) qubits and 2132 physical two-qubit gates, we still observe better-than-random performance.
  Finally, we show that adding active QEC cycles and increasing the repeat-until-success limit of state preparation subroutines can improve the performance of a quantum algorithm, thereby demonstrating critical capabilities of scalable FT quantum computation.
  Our results are enabled by a FT logical $T$ gate implementation with an infidelity of $2.7(4)\times10^{-3}$ and dynamic circuits with measurement-dependent feedback.
  Our work demonstrates near-break-even performance of complex, error-corrected algorithmic quantum circuits using only FT components.
\end{abstract}

\maketitle

\section{Introduction}

Realizing a broad class of high-impact applications of quantum computers at scale is widely anticipated to require quantum error correction (QEC) and fault tolerance.
Recent progress in quantum hardware has enabled demonstrations of below-threshold quantum memory~\cite{acharya2025quantum, he2025experimental}, with the logical error rate decreasing with code distance.
Furthermore, quantum primitives have been demonstrated below break-even~\cite{reichardt2024demonstration, paetznick2024demonstration}, with the logical error rate falling below the physical error rate for small error-correcting codes of fixed distance.

A central outstanding challenge is demonstrating fault-tolerant execution of complex algorithmic circuits below break-even, with the encoded circuit outperforming the same circuit executed directly using physical operations.
Realizing such a demonstration requires all components of the QEC stack to work together to support application execution.
While the success of recent component-level demonstrations suggests this milestone may be achievable on near-term devices, many obstacles remain that prevent straightforward combination of components in an end-to-end pipeline.

One obstacle is the number of different primitives required to execute quantum algorithms in their entirety, and the variety of options available for implementing these primitives.
This challenge is particularly acute in architectures that support qubit movement, such as trapped ions~\cite{pino2021demonstration, moses2023racetrack, ransford2025helios} and neutral atoms~\cite{bluvstein2024logical, reichardt2025faulttolerant, radnaev2025universal, bluvstein2025faulttolerant}, as the flexibility afforded by qubit movement allows for a broad selection of quantum error-correcting codes~\cite{bravyi2010tradeoffs, baspin2022quantifying, dai2025locality}, gate implementations~\cite{sayginel2025faulttolerant}, and approaches to fault tolerance~\cite{cain2024correlated, zhou2025lowoverhead, cain2025fast, ataides2025neural}.

We focus on the Steane code~\cite{steane1996multipleparticle}, which is one of the smallest and simplest error-correcting codes.
The development of specialized circuits for performing QEC~\cite{reichardt2020faulttolerant} and optimized gadgets for universal computation~\cite{goto2016minimizing} with the Steane code enabled early demonstrations of real-time fault-tolerant QEC~\cite{ryan-anderson2021realization} and the realization of a universal fault-tolerant gate set~\cite{postler2022demonstration}.
Further hardware improvements led to the achievement of break-even (that is, better-logical-than-physical) performance for select Clifford~\cite{paetznick2024demonstration} and non-Clifford~\cite{dasu2025breaking} gates.
At the same time, the Steane code has enabled \emph{application-level} error-corrected hardware benchmarking, notably including the three-qubit quantum Fourier transform~\cite{mayer2024benchmarking} and a quantum computation of molecular energies~\cite{yamamoto2025quantum}, mirroring the standard practice of application-level benchmarking of classical computing hardware~\cite{dai2019benchmarking}.
Due to the overheads of fully fault-tolerant execution, however, all demonstrations of logically encoded quantum algorithms (including Refs.~\cite{mayer2024benchmarking, yamamoto2025quantum}) have, to date, been performed only at the component level \cite{abobeih2022faulttolerant, postler2022demonstration, bluvstein2024logical, wang2024faulttolerant, lacroix2024scaling, bluvstein2025faulttolerant, dasu2026computing}, using non-fault-tolerant components \cite{mayer2024benchmarking, yamamoto2025quantum}, without error correction \cite{yamamoto2024demonstrating, hangleiter2025faulttolerant, reichardt2025faulttolerant}, or with a non-universal gate set \cite{reichardt2024demonstration}.

In this work, we demonstrate and benchmark the execution of error-corrected quantum algorithms using only fault-tolerant gadgets on trapped-ion quantum computers based on the quantum charge-coupled device (QCCD)~\cite{pino2021demonstration} architecture, namely \texttt{H2-1}, \texttt{H2-2}~\cite{moses2023racetrack}, and \texttt{Helios}~\cite{ransford2025helios}.
Our emphasis is on how these gadgets integrate into complete algorithms, rather than on the performance of any single gadget in isolation.
To this end, we first review and benchmark, at the component level, circuit gadgets that enable universal quantum computation with the Steane code.
Among the components that enable this demonstration, we introduce a fault-tolerant logical $T$ gate that outperforms the prior state of the art on hardware, achieving a fidelity of $0.960(6)$ for a Ramsey sequence with 16 $T$ gates on \texttt{H2-1}, corresponding to an infidelity of approximately $2.7(4)\times10^{-3}$ per $T$ gate.
For comparison, prior work \cite[Fig.~5, T method 2]{mayer2024benchmarking} demonstrated a non-fault-tolerant logical $T$ gate with infidelity $14\times10^{-3}$.
We additionally benchmark several methods for performing a QEC cycle, and find that a Steane-type gadget that teleports data onto a fresh code block outperforms a previously used flag-fault-tolerant protocol.

Putting these ingredients together, we execute the quantum approximate optimization algorithm (QAOA)~\cite{farhi2014quantum, blekos2024review} for small instances of the low autocorrelation binary sequences (LABS) problem~\cite{boehmer1967binary, schroeder1970synthesis} and a portfolio optimization problem instance, as well as the Harrow-Hassidim-Lloyd (HHL) algorithm~\cite{harrow2009quantum} for a structured toy problem, namely the Poisson equation on a lattice.
Notably, we observe improved QAOA performance with increasing QAOA depth and when compiling to more $T$ gates, the noisiest element of our logical instruction set.
Encoded circuits with few $T$ gates consistently perform comparably to their unencoded versions.
This near-break-even performance is made possible by both hardware improvements and the use of optimized fault-tolerant gadgets.

Moreover, we find that increasing the number of allowed attempts in repeat-until-success (RUS) subroutines has little effect on the performance of QAOA on \texttt{H2-2} and no effect on \texttt{Helios}, while reducing the post-selection discard rate to nearly zero, thereby demonstrating an essential capability for scalable fault-tolerant quantum computation.
Our HHL benchmarks similarly demonstrate the suppression of state preparation errors with RUS, and the improvement of logical circuit fidelity with active QEC cycles.
Through careful examination of HHL data, we identify dynamic program compilation artifacts as performance bottlenecks for the algorithm---highlighting the importance of application-level benchmarking to reveal system-level behaviors that may not manifest in component benchmarks.

The remainder of this paper is structured as follows.
We review the Steane code and circuit gadgets necessary for universal fault-tolerant quantum computation in \cref{sec:steane}, where we additionally introduce improved methods for applying a fault-tolerant $T$ gate.
We review the QAOA algorithm, the LABS problem, and binary portfolio optimization, as well as the HHL algorithm and a toy problem suitable for near-term implementation in \cref{sec:applications}.
We benchmark several Steane gadgets in isolation both in emulators and on hardware in \cref{sec:results_components}, and perform our application-level benchmarking in \cref{sec:results_applications}.
Finally, we close with a discussion in \cref{sec:discussion}.

\section{Fault-tolerant quantum computing with the Steane code}
\label{sec:steane}

We begin by reviewing the ingredients necessary for universal fault-tolerant quantum computation with the Steane code and the various options for implementing them on trapped-ion quantum processors with qubit movement.
For the purposes of this work, we say that a circuit gadget or protocol is fault tolerant if every mid-circuit error propagates to a correctable error at the end of the protocol.
Here mid-circuit errors include single-qubit errors after single-qubit gates, two-qubit errors after two-qubit gates, and single-qubit memory errors.
For simplicity, we do not explicitly consider the correction of leakage errors, though we note that (un-heralded) leakage errors are equivalent to Pauli errors when teleporting quantum data onto fresh qubits.
This operational definition of fault-tolerant gadgets is sufficient for the Steane code, which makes no promise of correcting more than one error between any two rounds of QEC.

We remark that the above notion of fault-tolerance is a property of an individual gadget.
Guaranteeing fault tolerance at the level of a full computation requires additional measures, such as interleaving QEC cycles to form extended rectangles~\cite{aliferis2005quantum} or correlated decoding~\cite{cain2024correlated, zhou2025lowoverhead}.
The algorithms we run are built entirely from fault-tolerant gadgets, but we do not make them fault tolerant in this stronger, whole-computation sense; we return to this point, and its consequences for our benchmarks, in \cref{sec:decoding}.
We refer the reader to Refs.~\cite{aliferis2005quantum, gottesman2009introduction} for more general and nuanced definitions and discussions of fault tolerance.

Throughout this work, we denote a single-qubit identity operator by $I = \op{0} + \op{1}$, a Pauli-$X$ operator by $X = \op{0}{1} + \op{1}{0}$, a Pauli-$Z$ operator by $Z = \op{0} - \op{1}$, and let $Y = i X Z$.
We denote the action of a single-qubit operator $O$ on physical qubit $j$ by $O_j$.
For any $n$-bit string $\bm b=(b_1, b_2, \cdots, b_n) \in \set{0,1}^n$, we let $O(\bm b) = \prod_{j:b_j=1} O_j$ denote the action of $O$ on the qubits indicated by $\bm b$.
To disambiguate operators that act on physical and logical degrees of freedom, in this section we use an overline to indicate that operators act on encoded logical qubits, as in $\overline{O}$, though we will drop this notation in \cref{sec:results} when there is no ambiguity.
Finally, we refer to the physical qubits that encode a logical qubit as the \emph{data qubits} of an error-correcting code.

\subsection{The Steane code}

\begin{figure}
  \includegraphics{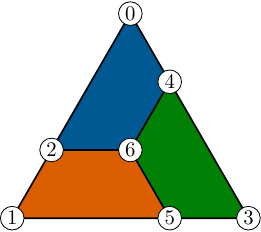}
  \caption{%
    A geometric representation of the Steane code.
    Nodes represent qubits.
    Colored cells identify stabilizer group generators, which are obtained by taking the product of either Pauli-$X$ or Pauli-$Z$ operators on the adjacent nodes, such as $X_1 X_2 X_6 X_5$.
    The product of either Pauli-$X$ or Pauli-$Z$ operators on any exterior edge of the triangle yields, respectively, a logical Pauli-$X$ or Pauli-$Z$ operator, such as $Z_0 Z_4 Z_3$.
  }
  \label{fig:steane_visual}
\end{figure}

The Steane code~\cite{steane1996multipleparticle} is the smallest Calderbank–Shor–Steane (CSS) quantum error-correcting code~\cite{steane1996multipleparticle, calderbank1996good} (as well as the smallest error-correcting topological and color code~\cite{bombin2013introduction, kubica2015universal}), and has code parameters $\llbracket 7, 1, 3 \rrbracket$, meaning that it encodes one logical qubit into 7 physical qubits and has code distance 3.
The Steane code can correct any single-qubit Pauli error, as well as two-qubit errors of mixed $X/Z$ type, such as $X_0 Z_4$.
Logical states of the Steane code are $+1$ eigenvectors of all stabilizers $s\in\mathcal{S}$, where the stabilizer group $\mathcal{S}$ consists of all products of stabilizer generators $g \in G = G_X\cup G_Z$, and the generator subgroup $G_P = \set{P(\bm b):\bm b\in\rows(h)}$ is defined in terms of the rows of the parity check matrix of the classical $[7, 4, 3]$ Hamming code,
\begin{align}
  h =
  \begin{pmatrix}
    0 & 0 & 0 & 1 & 1 & 1 & 1 \\
    0 & 1 & 1 & 0 & 0 & 1 & 1 \\
    1 & 0 & 1 & 0 & 1 & 0 & 1
  \end{pmatrix}.
  \label{eq:parity_check_matrix}
\end{align}
The logical $X$ and $Z$ operators of the Steane code are
\begin{align}
  \overline X &= \prod_{j=0}^6 X_j \simeq X_0 X_1 X_2,
  &&
  \overline Z &= \prod_{j=0}^6 Z_j \simeq Z_0 Z_1 Z_2,
\end{align}
where $\simeq$ denotes equality up to multiplication by stabilizers, and throughout this work we index the qubits of a Steane code block from 0 to 6.
\cref{fig:steane_visual} shows a geometric representation of the Steane code.

We note that there are different conventions in the literature regarding the ordering of qubits for the Steane code.
We use the ordering in (for example) Refs.~\cite{steane1996multipleparticle, goto2016minimizing, chamberland2019faulttolerant, heussen2023strategies, ibe2025measurementbased}, which differs from that in Refs.~\cite{ryan-anderson2018pecos, ryan-anderson2021realization, mayer2024benchmarking, yamamoto2025quantum}.
The map from qubit index in Refs.~\cite{ryan-anderson2018pecos, ryan-anderson2021realization, mayer2024benchmarking, yamamoto2025quantum} to that here and in Refs.~\cite{steane1996multipleparticle, goto2016minimizing, chamberland2019faulttolerant, heussen2023strategies, ibe2025measurementbased} is
$\set{4 \to 0, 5 \to 1, 6 \to 2, 3 \to 3, 2 \to 4, 1 \to 5, 0 \to 6}$.

\subsection{Clifford gates and initialization}
\label{sec:steane_clifford}

The Steane code has a transversal logical Hadamard gate, $\overline H = \prod_{j=0}^6 H_j$, and a transversal logical phase gate $\overline S = \prod_{j=0}^6 S_j^\dag$, where $H=\op{+}{0} + \op{-}{1}$ and $S=\op{0} + i\op{1}$.
These gates suffice to make all single-qubit Clifford gates transversal in the Steane code.
A logical $\overline{\mathrm{CNOT}}$ gate of the Steane code is implemented by broadcasting physical $\mathrm{CNOT} = \op{0}\otimes I + \op{1}\otimes X$ gates pairwise across the data qubits of two Steane code blocks.
That is, if logical qubits $\overline 0$ and $\overline 1$ are supported, respectively, on physical qubits $0,1,\cdots,6$ and $7,8,\cdots,13$, then $\overline{\mathrm{CNOT}}_{\overline 0,\overline 1} = \prod_{j=0}^6 \mathrm{CNOT}_{j,j+7}$.
The fact that these logical gates are transversal, meaning that every logical gate acts on at most one data qubit of a Steane code block, makes them automatically fault tolerant because there is no mechanism for a gate error on one data qubit in a Steane code block to propagate to another qubit within the same code block.
Transversality of the logical $H$ and $\mathrm{CNOT}$ implies that the logical controlled-$Z$ gate, $\mathrm{CZ} = \op{0}\otimes I + \op{1}\otimes Z$, is also transversal.

\begin{figure}
  \centering
  \includegraphics[scale=0.75]{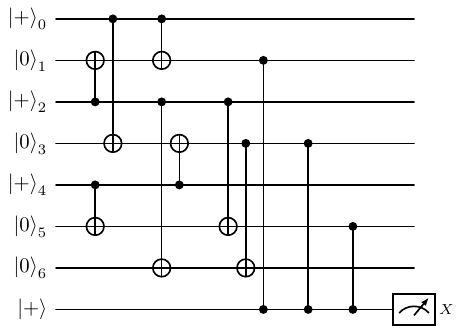}
  \caption{%
    Repeat-until-success protocol to fault-tolerantly prepare a logical $\ket{\overline{0}}$ state of the Steane code using CNOT and CZ gates.
    The bottom qubit is an ancilla (``flag'') qubit that is physically measured in the $X$ basis to read out the value of the logical $\overline{Z} \simeq Z_1 Z_3 Z_5$ operator.
    This protocol ``succeeds'' if the flag measurement outcome is 0 (that is, if the flag qubit is found to be in $\ket{+}$ upon measurement in the $X$ basis); otherwise, the qubits in this circuit are reset and the protocol is repeated until success.
  }
  \label{fig:prep_z}
\end{figure}

To initialize logical states of the Steane code, we use the repeat-until-success protocol for $\ket{\overline{0}}$-state preparation in Ref.~\cite{goto2016minimizing}, shown in \cref{fig:prep_z}, which uses an ancilla qubit to measure $\overline{Z}$ and discard faulty outcomes.
This circuit has the same two-qubit gate count as, for example, fault-tolerant circuits discovered with reinforcement learning~\cite{zen2025quantum}.
Other Pauli eigenstates of the Steane code, such as $\ket{\overline{+}}$, can be prepared by initializing $\ket{\overline{0}}$ and applying a logical Clifford gate.

\subsection{Measurement and decoding}
\label{sec:steane_measurement}

A logical $Z$-basis measurement of a Steane code block is performed by measuring all data qubits in the $Z$ basis, which yields a 7-bit string $\bm b_{\mathrm{measured}}$.
This bitstring may generally contain errors, which can be diagnosed by the \emph{error syndrome} $\bm s = h \cdot \bm b_{\mathrm{measured}}$ (mod~2).
In turn, the syndrome $\bm s$ needs to be \emph{decoded} to infer a correction $\bm b_{\mathrm{measured}}\to\bm b_{\mathrm{corrected}}$.
In the case of the Steane code, a non-zero syndrome $\bm s$ corresponds to a bit-flip error on bit $\integer(\bm s)-1$ of $\bm b_{\mathrm{measured}}$ (that is, the measurement bit associated with qubit $\integer(\bm s)-1$), where $\integer(\bm s)$ is the integer encoded by the bitstring $\bm s$, with (for example) $\integer((0,0,1))=1$ and $\integer((0,1,0))=2$.
This simple decoding procedure is a consequence of the facts that
\begin{enumerate*}[label=(\alph*)]
\item only single-qubit errors are correctable in the Steane code, and
\item the nonzero bitstring $\bm v$ is the $(\integer(\bm v)-1)$-th column of $h$ in \cref{eq:parity_check_matrix}.
\end{enumerate*}
After correcting the error inferred from $\bm s$ by flipping the corresponding bit in $\bm b_{\mathrm{measured}}$ to obtain $\bm b_{\mathrm{corrected}}$, the parity of $\bm b_{\mathrm{corrected}}$ is the logical $Z$-basis measurement outcome.
A logical $X$-basis measurement can be performed identically by measuring all data qubits in the $X$-basis, and using the same decoding procedure.

\subsection{Non-Clifford gates}
\label{sec:steane_non_cliff}

\begin{figure*}
  \subfloat[%
    Logical $\ket{\overline{H}}$ state preparation by directly encoding a physical $\ket{H}$ state.
  ]
  {\label{fig:prep_h_old}\includegraphics[scale=0.75]{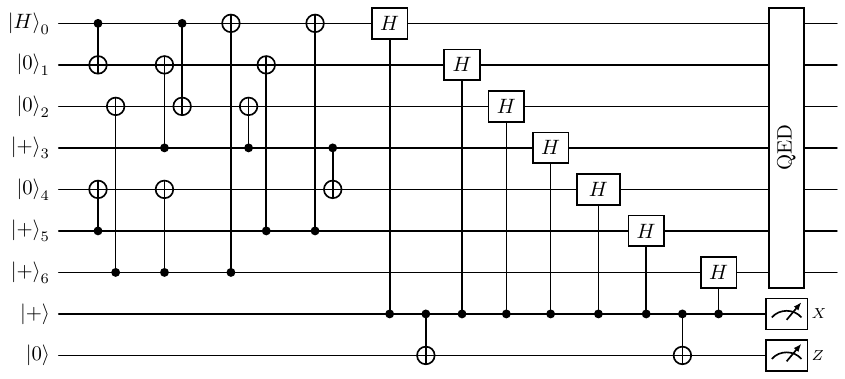}} \\
  \subfloat[%
    Specialized $\ket{\overline{H}}$-state preparation circuit.
  ]
  {\label{fig:prep_h_new}\includegraphics[scale=0.75]{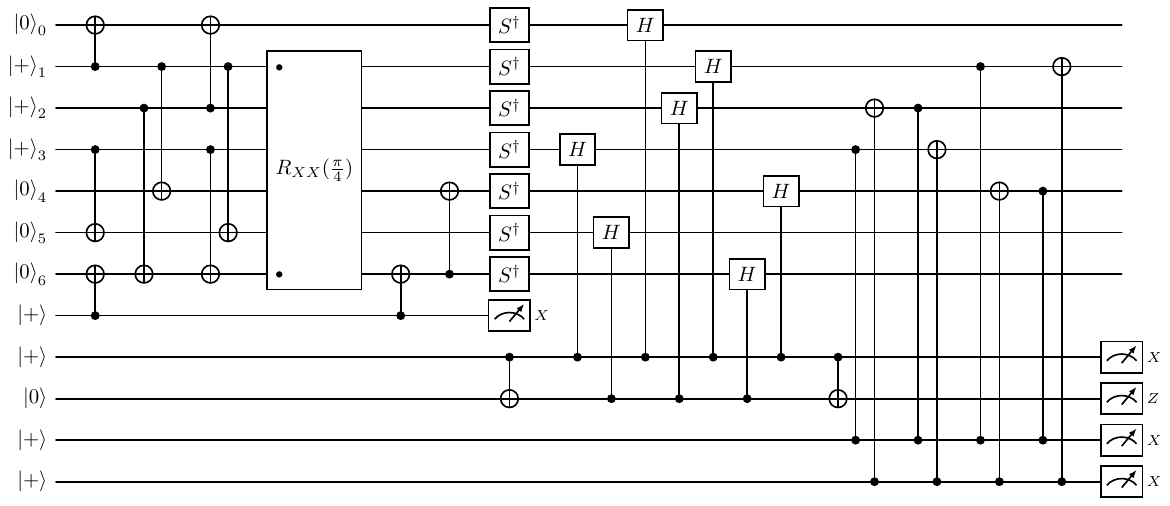}}
  \caption{%
    Repeat-until-success protocols to fault-tolerantly prepare $\ket{\overline{H}}$.
    Measurements are performed in the $Z$ or $X$ basis as indicated.
    \textbf{\sref{fig:prep_h_old}} A circuit that non-fault-tolerantly encodes a single-qubit $\ket{H}$ state into $\ket{\overline{H}}$, fault-tolerantly measures $\overline{H}$, and runs a full round of QED that consists of six stabilizer measurements.
    \textbf{\sref{fig:prep_h_new}} A specialized circuit to fault-tolerantly prepare $\ket{\overline{H}}$ using a shorter non-fault-tolerant $\ket{\overline{H}}$ preparation circuit, two qubits to flag mid-circuit errors, a fault-tolerant measurement of $\overline{H}$, and two stabilizer measurements.
    Here $R_{XX}(\theta) = e^{-i\theta X\otimes X/2}$ is a two-qubit Pauli rotation on the qubits indicated by dots ($\bullet$).
    The circuit here has been expanded for clarity, but can be implemented with a two-qubit gate depth of 9.
  }
  \label{fig:prep_h}
\end{figure*}

Universal quantum computation requires non-Clifford gates.
The presence of a transversal logical Hadamard gate, $\overline{H}$, enables a simple strategy to fault-tolerantly prepare a non-Clifford state: the $+1$ eigenstate of $\overline{H}$, which we denote $\ket{\overline{H}}$.
The basic idea is to first prepare $\ket{\overline{H}}$ non-fault-tolerantly, and then perform a fault-tolerant measurement of $\overline{H}$ and some quantum error detection (QED) to project out preparation errors.
This is a repeat-until-success protocol: nonzero measurement outcomes trigger a reset and restart of the protocol.
\cref{fig:prep_h}\sref{fig:prep_h_old} shows a straightforward implementation of this protocol, which has featured in prior works~\cite{goto2016minimizing, postler2022demonstration, mayer2024benchmarking}.

We construct an improved circuit to fault-tolerantly prepare the $\ket{\overline{H}}$ state in \cref{fig:prep_h}\sref{fig:prep_h_new}.
This circuit uses the non-fault-tolerant $\ket{\overline{T}}$-state preparation circuit in Eq.~B13 of Ref.~\cite{yamamoto2025quantum}, which (similarly to the circuit in Fig.~2 of Ref.~\cite{ibe2025measurementbased}) requires only 9 two-qubit gates, as opposed to the 11 two-qubit gates required for the encoding circuit in \cref{fig:prep_h}\sref{fig:prep_h_old}.
Similar circuits can be constructed by propagating representations of the $\overline{X}$ operator (that may differ by stabilizer factors) backward through a $\ket{\overline{0}}$-state preparation circuit to find locations at which the back-propagated $\overline{X}$ is a Pauli string with weight 2.
Inserting a pair-Pauli rotation gate at such a location makes the circuit prepare the state $R_{\overline{X}}(\theta)\ket{\overline{0}} = \exp(-i\frac{\theta}{2}\overline{X}) \ket{\overline{0}} = \cos(\frac{\theta}{2}) \ket{\overline{0}} + \sin(\frac{\theta}{2}) \ket{\overline{1}}$.
In addition, the specialized circuit in \cref{fig:prep_h}\sref{fig:prep_h_new} uses one qubit to flag mid-circuit errors, and measures only two stabilizers for QED after the measurement of $\overline{H}$.
We provide a detailed discussion of \cref{fig:prep_h}\sref{fig:prep_h_new} and its construction in \cref{sec:prep_h}.

We note that Ref.~\cite[Fig.~6]{chamberland2019faulttolerant} and Ref.~\cite[Fig.~11]{heussen2023strategies} provide deterministic protocols to fault-tolerantly prepare a $\ket{\overline{H}}$ state of the Steane code.
These strategies use additional flag qubits to diagnose mid-circuit errors.
We focus on repeat-until-success strategies in this work, and refer to the maximum number of attempts allowed in a repeat-until-success subroutine before discarding the entire circuit and restarting a computation as the ``repeat-until-success (RUS) limit''.

\begin{figure}
  \centering
  \begin{tabular}{cc}
    \subfloat[\texttt{T-bare}]
    {\label{fig:t_gate_bare}\includegraphics[scale=0.75]{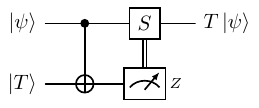}}
    &
    \subfloat[\texttt{T-swap}]
    {\label{fig:t_gate_swap}\includegraphics[scale=0.75]{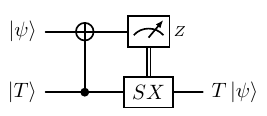}}
    \\
    \subfloat[\texttt{Steane-bare}]
    {\label{fig:qec_steane_bare}\includegraphics[scale=0.75]{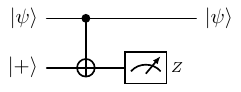}}
    &
    \subfloat[\texttt{Steane-swap}]
    {\label{fig:qec_steane_swap}\includegraphics[scale=0.75]{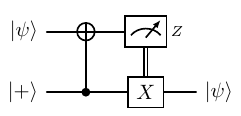}}
  \end{tabular}
  \caption{%
    \sref{fig:t_gate_bare},\sref{fig:t_gate_swap}:
    Two gadgets that consume a $\ket{T} = T\ket{+}$ state to apply a $T$ gate to an arbitrary single-qubit state $\ket\psi$.
    At the level of circuit identities, these gadgets can be generalized by replacing $T\to R_Z(\theta)$ and $S\to R_Z(2\theta)$ for any angle $\theta$, but the resulting circuits may not have fault-tolerant implementations.
    \sref{fig:qec_steane_bare},\sref{fig:qec_steane_swap}:
    Analogous gadgets (with $\theta=0$) that implement a fault-tolerant error detection cycle for bit-flip ($X$) errors.
    For \sref{fig:qec_steane_bare}, detected errors can be either actively corrected or tracked in software.
    For \sref{fig:qec_steane_swap}, the conditional logical $X$ correction can similarly be tracked in software by updating a Pauli frame.
    Gadgets \sref{fig:t_gate_swap} and \sref{fig:qec_steane_swap} have the benefit of mitigating leakage errors in $\ket\psi$ and correcting physical $X$ errors passively, as these errors do not propagate to the output state of the gadget.
    Changing bases as $Z\leftrightarrow X$, $\ket+\leftrightarrow\ket0$, and flipping the orientation of the CNOT yields gadgets for fault-tolerantly implementing a $T_X = HTH$ gate and detecting phase-flip errors.
  }
  \label{fig:t_gate_and_qec}
\end{figure}

After preparing $\ket{\overline{H}}$, this state can be rotated into the more commonly used $\ket{\overline{T}} = \overline{T} \ket{\overline{+}}$ state as $\ket{\overline{T}} = \overline{H}\overline{S}\ket{\overline{H}}$.
Here $T = \op{0} + e^{i\pi/4}\op{1}$.
The $\ket{\overline{T}}$ state can in turn be consumed by one of the gadgets in \cref{fig:t_gate_and_qec}\sref{fig:t_gate_bare} or \cref{fig:t_gate_and_qec}\sref{fig:t_gate_swap} to apply a $\overline{T}$ gate and thereby complete a Clifford+$T$ gate set for the Steane code.

\subsection{Quantum error correction}
\label{sec:qec}

Finally, we discuss strategies for performing QEC and QED in the Steane code.
In principle, to perform QEC one needs to measure a complete set of stabilizers that generate the code's stabilizer group.
The syndrome (bitstring) obtained from these stabilizer measurements can then be decoded to infer a data qubit error.
In principle, each stabilizer can be measured using a gadget similar to that in \cref{fig:pauli_measurement_gadget}.
However, some care is needed to make sure that stabilizer measurement gadgets are themselves fault tolerant.
Moreover, redundant syndrome measurements are required to diagnose measurement errors.

\begin{figure}
  \includegraphics[scale=0.75]{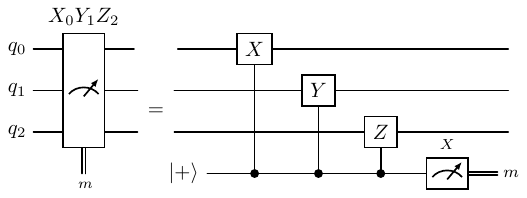}
  \caption{%
    Non-fault-tolerant gadget to measure the Pauli string $X_0 Y_1 Z_2$, writing the result to the bit $m$.
    Single-qubit errors in the middle of the decomposition on the right can propagate to higher-weight errors on qubits $q_0,q_1,q_2$.
  }
  \label{fig:pauli_measurement_gadget}
\end{figure}

The paradigm of \emph{flag fault tolerance}~\cite{chamberland2018flag, chao2020flag} provides a general scheme for fault-tolerant stabilizer measurements.
This paradigm typically makes careful use of additional flag qubits whose measurement outcomes provide sufficient information to diagnose mid-circuit errors.
In the case of the Steane code, Ref.~\cite{reichardt2020faulttolerant} developed a specialized flag-fault-tolerant syndrome measurement protocol that uses some syndrome measurement qubits as flags for others, altogether requiring only six ancilla qubits to measure six stabilizers (and requiring only three ancilla qubits at a time).
If an initial round of flagged stabilizer measurements yields a nontrivial syndrome, this protocol triggers a second round of stabilizer measurements, and the combined syndromes from the two measurement rounds suffice to diagnose and correct individual mid-circuit errors in a fault-tolerant manner.
This protocol was experimentally implemented in Ref.~\cite{ryan-anderson2021realization}.
We review and provide additional details about the flag-fault-tolerant QEC protocol used in this work, including explicit circuits, in \cref{sec:flag_qec}.

An alternative general strategy for fault-tolerant QEC is Steane error correction (note the naming collision: Steane error correction need not be performed with the Steane code).
Steane QEC can be thought of as ``copying'' the physical errors in one code block onto an ancilla code block---crucially, without affecting the logical state of either code block---and then measuring out the ancilla code block to diagnose the errors.
This strategy is shown in \cref{fig:t_gate_and_qec}\sref{fig:qec_steane_bare}, and can be recovered from a generalization of the $T$ gate teleportation gadget in \cref{fig:t_gate_and_qec}\sref{fig:t_gate_bare}.
Physical $X$ errors on the top (logical) qubit propagate through the transversal CNOT to physical $X$ errors on the bottom (logical) qubit, which in turn flip $Z$-basis measurements of the ancilla, thereby revealing the location of the $X$ errors.
However, the CNOT does not affect the logical states of these qubits: $\mathrm{CNOT}(\ket\psi\otimes\ket+) = \ket\psi\otimes\ket+$.
Preparing the ancilla qubit in $\ket{\overline{0}}$, flipping the orientation of the CNOT, and measuring the ancilla qubit in the $X$ basis similarly diagnoses physical phase-flip ($Z$) errors.
Diagnosed errors can be either actively corrected or virtually tracked in software.
We refer to the execution of \cref{fig:t_gate_and_qec}\sref{fig:qec_steane_bare} followed by its $Z$-correction variant as the \texttt{Steane-bare} strategy for QEC in this work.

A teleportation-like variant of Steane QEC, shown in \cref{fig:t_gate_and_qec}\sref{fig:qec_steane_swap}, can similarly be recovered from a generalization of the $T$ gate teleportation gadget in \cref{fig:t_gate_and_qec}\sref{fig:t_gate_swap}.
This strategy corrects errors \emph{passively}: physical $X$ errors on the top (logical) qubit do not propagate to the bottom (logical) qubit, and instead are ``decoded away'' in the logical $Z$-basis measurement of the top qubit.
As with ordinary Steane QEC, appropriately swapping $X$ and $Z$ bases in \cref{fig:t_gate_and_qec}\sref{fig:qec_steane_swap} yields a gadget for correcting phase-flip ($Z$) errors.
We refer to the execution of \cref{fig:t_gate_and_qec}\sref{fig:qec_steane_swap} followed by its $Z$-correction variant as the \texttt{Steane-swap} strategy for QEC in this work.

\begin{figure}
  \includegraphics[scale=0.75]{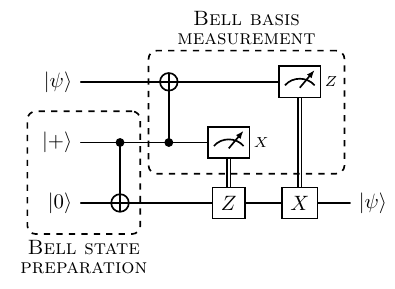}
  \caption{Standard teleportation circuit for Knill QEC.}
  \label{fig:qec_knill}
\end{figure}

Finally, QEC can also be performed with full quantum state teleportation, which is commonly referred to as Knill QEC.
A standard circuit for this strategy is shown in \cref{fig:qec_knill}.
A minor modification of this strategy, in which the second logical CNOT is performed before the first, reduces Knill QEC to a parallelized version of (two complementary rounds of) \texttt{Steane-swap} QEC.

\section{Application benchmarks}
\label{sec:applications}

Here we describe the algorithms that we use for application-level benchmarking of the Steane code on the \texttt{H2-1}, \texttt{H2-2}~\cite{moses2023racetrack}, and \texttt{Helios}~\cite{ransford2025helios} trapped-ion quantum computers.

\subsection{QAOA}
\label{sec:qaoa_background}

We first consider the quantum approximate optimization algorithm (QAOA)~\cite{farhi2014quantum, blekos2024review}, which takes as input
\begin{enumerate*}[label=(\alph*)]
\item a \emph{cost Hamiltonian} $H_C$ that encodes an $N$-bit optimization problem, and
\item two real-valued vectors $\bm\beta,\bm\gamma\in\mathbb{R}^p$,
\end{enumerate*}
and prepares the state $\ket{\bm\beta,\bm\gamma} = U_{\mathrm{QAOA}}(\bm\beta,\bm\gamma) \ket{+}^{\otimes N}$, where $p$ is referred to as the \emph{QAOA depth} and the unitary
\begin{align}\label{eq:standard_qaoa_ansatz}
  U_{\mathrm{QAOA}}(\bm\beta,\bm\gamma)
  = \prod_{j=1}^p e^{-i\beta_j H_M} e^{-i\gamma_j H_C}
\end{align}
additionally depends on a choice of \emph{mixing Hamiltonian} $H_M$ that is typically taken to be a transverse field, $H_M = \frac12 \sum_j X_j$.
The output of QAOA is a list of bitstrings obtained by repeatedly preparing and measuring the state $\ket{\bm\beta,\bm\gamma}$ in the computational basis.
The \emph{success probability} of a QAOA state $\ket{\bm\beta,\bm\gamma}$ is the probability of measuring the optimal solution to the optimization problem, or the ground state of $H_C$.

The cost Hamiltonian $H_C$ for a cost function $C:\set{0,1}^N\to\mathbb{R}$ is typically constructed by energetically penalizing the bitstrings according to their cost, for example with $H_C = \sum_{\bm x\in\set{0,1}^N} C(\bm x)\op{\bm x}$.
The QAOA circuit is equivalent, up to a choice of $\bm\beta, \bm\gamma$, to the Trotterized evolution of an Ising model with a transverse field.

Extensive theoretical and numerical evidence suggests that QAOA can provide a polynomial~\cite{boulebnane2024solving, shaydulin2024evidence, boulebnane2025evidence, apte2026quantum} and, in some restricted cases, exponential~\cite{farhi2025lower, montanaro2025quantum} speedup over state-of-the-art classical algorithms.
Though QAOA can be viewed as a variational algorithm in which the parameters $\bm\beta,\bm\gamma$ are optimized inside a classical control loop, for some problem classes it is possible to obtain fixed (instance-independent) parameters~\cite{shaydulin2024evidence, sureshbabu2024parameter, boulebnane2024solving, farhi2022quantum, wurtz2021fixedangle, farhi2025lower, he2025parameter, he2025nonvariational} that nonetheless yield an advantage over classical state-of-the-art methods, as is the case for the problems considered in this work~\cite{shaydulin2024evidence, sureshbabu2024parameter}.
Consequently, in this work we compile QAOA with fixed parameters $\bm\beta,\bm\gamma$.

\subsubsection{LABS}
\label{sec:labs}

The first problem that we tackle with QAOA is the low-autocorrelation binary sequences (LABS) problem~\cite{boehmer1967binary, schroeder1970synthesis}, which has applications in communications engineering~\cite{schroeder1970synthesis, golay1977sieves}, and for which there exists evidence of a scaling advantage for QAOA over state-of-the-art classical methods~\cite{shaydulin2024evidence}.
The objective of LABS is to find a binary sequence $\bm s = (s_1, s_2, \cdots, s_N) \in\set{+1,-1}^N$ that minimizes the autocorrelation energy
\begin{align}
  E_{\mathrm{LABS}}(\bm s) = \sum_{k=1}^{N-1} A_k(\bm s)^2,
  &&
  A_k(\bm s) = \sum_{j=1}^{N-k} s_j s_{j+k}.
  \label{eq:E_LABS}
\end{align}
The QAOA cost Hamiltonian $H_{\mathrm{LABS}}$ for LABS is obtained by promoting each variable $s_j$ in \cref{eq:E_LABS} to a single-qubit Pauli-$Z$ operator, $Z_j$.
The LABS problem becomes classically intractable to solve exactly at $\approx 100$ binary variables, making it a promising candidate for demonstrating quantum advantage in optimization~\cite{zhang2025new, boskovic2017lowautocorrelation, boskovic2016github, packebusch2016low}.
Ref.~\cite{shaydulin2024evidence} provides optimized QAOA parameters for LABS for small $N$, and instance-independent parameters for large $N$, which we use in this work.

\subsubsection{Portfolio optimization}

The second problem we tackle with QAOA is portfolio optimization.
We consider the simplest version of this problem, namely, binary unconstrained mean-variance portfolio optimization, which takes the form of a quadratic unconstrained binary optimization problem.
The goal is to find a binary vector $\bm s = (s_1, s_2, \cdots, s_N) \in \{0,1\}^N$ that indicates whether a given asset is included or excluded from a portfolio to minimize the objective function
\begin{align}
  C(\bm s) = q \sum_{ij} W_{ij} s_i s_j - \sum_i \mu_i s_i,
  \label{eq:portfolio_opt}
\end{align}
where $\bm\mu \in \mathbb{R}^{N}$ denotes the vector of expected returns for the assets, $W \in \mathbb{R}^{N \times N}$ is the covariance matrix between the $N$ assets, and $q > 0$ is a risk factor that balances risk and return in the objective function.
The corresponding QAOA cost Hamiltonian $H_C$ is obtained by promoting each variable $s_j$ in \cref{eq:portfolio_opt} to the single-qubit projector $\op{1}_j = \frac12 (1 - Z_j)$.
Letting $C_{\text{min}} = \min_{\bm s} C(\bm s)$ and $C_{\text{max}} = \max_{\bm s} C(\bm s)$ denote the extrema of $C$, the approximation ratio achieved by a QAOA state $\ket{\bm\beta,\bm\gamma}$ is
\begin{align}
  \mathrm{AR}(\bm\beta,\bm\gamma)
  = \frac{\braket{\bm\beta,\bm\gamma|H_C|\bm\beta,\bm\gamma} - C_{\min}}{C_{\max} - C_{\min}}.
\end{align}
Portfolio optimization problem instances are defined by a choice of expected returns $\bm\mu$ and covariance matrix $W$ for a basket of assets, which we generate from historical equity price data using \texttt{qokit}~\cite{2025qokit}.

\subsection{HHL}
\label{sec:hhl}

The Harrow-Hassidim-Lloyd (HHL) algorithm~\cite{harrow2009quantum} was proposed for solving a quantum analog of a traditional linear system of equations, and can be stated as follows.
Given an error tolerance $\epsilon > 0$ and appropriate access to a vector $\bm b \in \mathbb{C}^M$ and matrix $A \in \mathbb{C}^{M\times N}$ for which $M \leq N$ and the operator norm $\norm{A} \leq 1$, prepare a quantum state $\ket{\bm y} \in \mathbb{C}^N$ that is $\epsilon$-close to a solution $\bm x_\star$ of $A\bm x = \bm b$, meaning $\norm{\ket{\bm y} - \ket{\bm x_\star}} < \epsilon$ for the normalized state $\ket{\bm x_\star}$ that is proportional to $\bm x_\star$ (where $\norm{\ket\psi} = \sqrt{\braket{\psi|\psi}}$).
This task is commonly referred to as the quantum linear systems problem (QLSP).

Various algorithmic improvements have extended the original approach of HHL~\cite{childs2017quantum, chakraborty2019power, an2022quantum}, eventually achieving a runtime that is asymptotically optimal in the number of queries to $A$ and $\bm b$, namely $\Theta(\kappa\log(1/\epsilon))$, where $1/\kappa$ is the smallest nonzero singular value of $A$~\cite{costa2022optimal}.
Altogether, algorithms for the QLSP with sparse matrices (that is, matrices whose maximum number of nonzero entries per row and column is independent of $M$ and $N$) have a runtime that is $\mathcal{O}(\kappa\log(1/\epsilon)\polylog(M,N))$.
Given the prevalence of matrix inversion as a subroutine in classical algorithms and the $\poly(M,N) \rightarrow \polylog(M,N)$ improvement for QLSP, it is no surprise that QLSP solvers have been used as components in proposed quantum algorithms targeting various problems, including differential equations~\cite{montanaro2016quantum, liu2021efficient, linden2022quantum}.

While the potential applications of QLSP solvers may seem vast at face value, there are many technical barriers to achieving end-to-end speedups for classical problems.
For example, QLSP solvers require efficient access to entries of $A$ and $\bm b$, and only provide an exponential speedup when the condition number $\kappa$ is $\mathcal{O}(\polylog(N,M))$.
Moreover, the output of a QLSP solver is a state $\ket{\bm y}$ that is stored in quantum memory.
Classical state readout, or the inference of the full classical vector $\bm y$ from the quantum state $\ket{\bm y}$, wipes out any quantum speedup of a QLSP solver.
For this reason, a QLSP solver should not be thought of as a direct solver for a classical linear system of equations, but rather as a routine to prepare the quantum state $\ket{\bm y}$.
The state $\ket{\bm y}$ may in turn be used to measure specific observables of interest, or as an input to other quantum subroutines; see Refs.~\cite{kerenidis2016quantum, liu2021efficient, babbush2023exponential, bravyi2025quantum} for examples of classical problems that benefit from this use of a QLSP.

\begin{figure*}
  \subfloat[%
    Sketch of the HHL algorithm for an $N\times N$ matrix $A$.
    Here $n=\lceil\log_2 N\rceil$, and the choice of $m$ limits the achievable precision of eigenvalue estimation and inversion.
    The last gate on the bottom qubit indicates post-selection on a $0$ measurement outcome.
  ]
  {\label{fig:circuit_hhl_sketch}
    \begin{minipage}{\textwidth}
      \centering
      \includegraphics[scale=0.75]{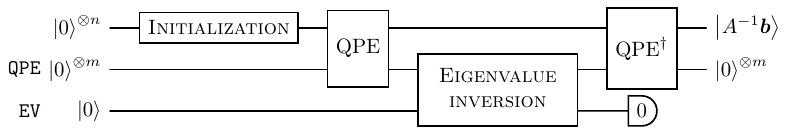}
    \end{minipage}
  } \\
  \subfloat[%
    Small-scale implementation of the HHL algorithm for the Poisson equation, with $m=n=2$.
  ]
  {\label{fig:hhl_poisson}\includegraphics[scale=0.75]{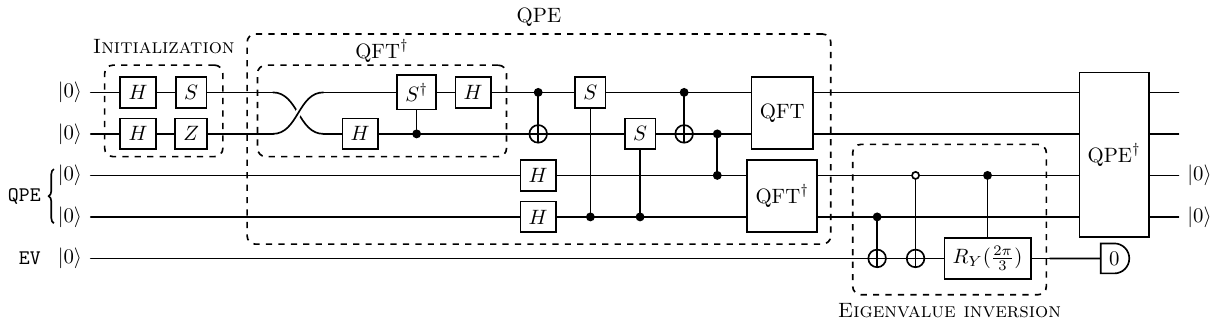}}
  \caption{%
    Logical circuits for the HHL algorithm.
    The HHL algorithm includes an explicit measurement and post-selection on a $\ket{0}$ state of the ``eigenvalue'' (\texttt{EV}) qubit.
    In the absence of errors, the remaining (\texttt{QPE}) ancillas are guaranteed to end in $\ket{0}$ states, which can be leveraged to perform logical error detection.
  }
  \label{fig:circuit_hhl}
\end{figure*}

We consider the original HHL algorithm in this work, and restrict to the case of Hermitian matrices $A$.
Though asymptotically suboptimal, the original HHL algorithm
\begin{enumerate*}[label=(\alph*)]
\item still uses important quantum subroutines like the quantum Fourier transform and quantum phase estimation (QPE)~\cite{kitaev1995quantum}, and
\item can be distilled into a small enough circuit for execution on current quantum hardware, unlike more sophisticated QLSP solvers.
\end{enumerate*}

At a high level, the HHL algorithm is as follows (see also \cref{fig:circuit_hhl}\sref{fig:circuit_hhl_sketch} for a sketch).
\begin{enumerate}
  \item\label{step:initialization} \textbf{Initialization}:
    Prepare the state $\ket{\bm b} = \bm{b} / \sqrt{\bm{b}^\dag \bm{b}}$.
  \item\label{step:QPE} \textbf{Eigenvalue estimation}:
    Denoting the eigenvalues of $A$ by $\lambda_j$, the corresponding eigenvectors by $\ket{e_j}$, and defining $b_j = \braket{e_j|\bm b}$, implement the transformation
    \begin{align}
      \ket{\bm b} = \sum_j b_j \ket{e_j} \to \sum_j b_j \ket{e_j} \otimes \ket{\lambda_j},
    \end{align}
    which can be done using QPE~\cite{kitaev1995quantum}.
  \item\label{step:inversion} \textbf{Eigenvalue inversion}:
    Letting $\tilde\lambda_j^{-1} = C / \lambda_j$ with some constant $C$ for which all $\tilde\lambda_j^{-1} \le 1$ (and defining $\tilde\lambda_j^{-1} = 0$ whenever $\lambda_j=0$), coherently implement the transformation
    \begin{align}
      \ket{\lambda_j}
      \to \ket{\lambda_j} \otimes
      \sp{\tilde\lambda_j^{-1} \ket{0} + \sqrt{1 - \tilde\lambda_j^{-2}} \ket{1}}.
    \end{align}
  \item\label{step:QPE_inv} \textbf{Eigenvalue uncomputation}:
    Run the inverse of step \ref{step:QPE} to uncompute the eigenvalue register, arriving at the state
    \begin{align}
      \sum_j b_j \ket{e_j} \otimes
      \sp{\tilde\lambda_j^{-1} \ket{0} + \sqrt{1 - \tilde\lambda_j^{-2}} \ket{1}}.
    \end{align}
  \item \textbf{Post-selection}:
    Measure the last qubit register and post-select on a measurement outcome of $0$, thereby preparing the state
    \begin{align}
      \ket{A^{-1}\bm b} \propto \sum_j \frac{b_j}{\lambda_j} \ket{e_j},
    \end{align}
    which solves the QLSP induced by $(A, \bm b)$.
\end{enumerate}
The details of step \ref{step:initialization} depend on the access model for $\bm b$.
Steps \ref{step:QPE} and \ref{step:QPE_inv} are implemented by well-studied QPE circuits~\cite[Chapter 5]{nielsen2010quantum} (see also \cref{fig:qpe_full} of \cref{sec:hhl_details}).
Step \ref{step:inversion} typically involves quantum arithmetic to map $\ket{\tilde\lambda^{-1}} \to \ket\eta = \ket{\arcsin\tilde\lambda^{-1}}$~\cite[Section 6]{haner2018optimizing}, and then coherently apply (for example) a rotation $R_Y(2\eta)$ to an ancilla qubit that is initially in $\ket{1}$.
More generally, every step above can be relaxed to implement a sufficiently close approximation of the desired transformation.

In full generality, HHL relies on quantum arithmetic and Hamiltonian simulation subroutines that may require more advanced quantum computers than are currently available.
As a benchmark, we therefore consider using HHL to solve a structured toy problem that is amenable to near-term implementation.
Specifically, we consider using HHL to solve the QLSP for the Poisson equation $\nabla^2 \bm x = \bm b$ on a translationally invariant one-dimensional lattice of $L=4$ points with periodic boundary conditions, for which
\begin{align}
  \nabla^2 = S_+ + S_- - 2\,\mathds{1},
\end{align}
where $S_\pm = \sum_{j=0}^{L-1} \op{j\pm1\,\mathrm{mod}\,L}{j}$ is a shift matrix and $\mathds{1}$ is the identity matrix.
We construct a circuit implementing HHL for this QLSP with $\bm b = (1, -1, i, -i)$ in \cref{sec:hhl_details}, altogether arriving at the circuit in \cref{fig:circuit_hhl}\sref{fig:hhl_poisson}.

\subsection{Logical circuit compilation}
\label{sec:compilation}

QAOA circuits rely on continuously parameterized rotation gates, which cannot be implemented transversally due to the Eastin-Knill theorem~\cite{eastin2009restrictions}.
We therefore compile QAOA circuits to a Clifford+$T$ gate set by
\begin{enumerate*}[label=(\alph*)]
\item mapping all rotations of the form $e^{-i\theta Z^{\otimes k}/2}$ to single-qubit $R_Z(\theta) = e^{-i\theta Z/2}$ gates with CNOT ladders, and
\item using \texttt{trasyn}~\cite{hao2025reducing} to decompose every $R_Z(\theta)$ gate into a sequence of Clifford+$T$ gates approximating it.
\end{enumerate*}
\texttt{trasyn} is a compiler that, given a Clifford+$R_Z$ circuit and a $T$ gate infidelity, returns a Clifford+$T$ circuit that tries to balance both synthesis error (that is, the fidelity of the decomposed circuit with respect to the original circuit, in the absence of noise) and gate error (determined by the number of $T$ gates).
In principle, a detailed and precise noise model of quantum hardware may be used to find an optimal decomposition of a Clifford+$R_Z$ circuit into Clifford+$T$.
In practice, \texttt{trasyn} uses a Pauli noise model on the logical level, which does not account for how physical errors (including coherent and idling errors particular to the QCCD architecture) affect logical performance.
We therefore tune the $T$ gate infidelity provided to \texttt{trasyn} in order to arrive at circuits with the $T$ gate counts in \cref{sec:results}.
An interesting direction for future work is to investigate noise-aware logical circuit compilation that is faithful to the noise profile of a QCCD architecture.

In the case of the HHL circuit in \cref{fig:circuit_hhl}\sref{fig:hhl_poisson}, each controlled-$S$ gate can be decomposed exactly into Clifford+$T$ using three $T$ gates, while the controlled-$R_Y(\frac{2\pi}{3})$ can be approximated to a channel fidelity of $\sim0.966$ using only two $T$ gates (see \cref{sec:hhl_details}).

After an initial decomposition into a Clifford+$T$ gate set, we further compile all logical benchmarking circuits by cycling through rounds of the following:
\begin{enumerate}
  \item Merging single-qubit gates (commuting them past two-qubit gates whenever possible).
  \item The \texttt{GreedyPauliSimp} compilation pass in \texttt{pytket}~\cite{sivarajah2020t|ket}, which implements the Pauli frame graph and Clifford circuit synthesis methods from Refs.~\cite{paykin2023pcoast, schmitz2024graph}.
  \item The \texttt{SimplifyInitial} and \texttt{SimplifyMeasured} compilation passes in \texttt{pytket} to simplify gates at the beginning and end of the circuit, which act on known initial states or precede measurements.
\end{enumerate}
This cycle of compilation passes is repeated until either the circuit reaches a fixed point or ten compilation cycles are applied, at which point the representation of the circuit throughout compilation that had the fewest $T$ gates (which is typically---but not always---the final circuit) is taken to be the compiled circuit.
As a final optimization, we identify cases in which a $\ket{0}$ state is acted on by single-qubit Clifford gates followed by a $T$ gate, and replace such cases by an initial $\ket{T}$ state followed by an appropriate sequence of Clifford gates, thereby eliminating a $T$ gate injection.
The compiled logical HHL circuit that we execute on \texttt{Helios} is provided in \cref{fig:hhl_compiled} of \cref{sec:hhl_details}.

\section{Results}
\label{sec:results}

We now benchmark components used for fault-tolerant quantum computing with the Steane code, before using the best methods identified in our component benchmarks for application-level benchmarking of small-scale quantum algorithms.
For brevity, from here onward we drop the overline notation for logical operators, $\overline{O}\to O$, as our discussions will concern only logical degrees of freedom unless explicitly stated otherwise.

\subsection{Component benchmarking}
\label{sec:results_components}

\subsubsection{$\ket{H}$-state preparation}

\begin{figure}
  \includegraphics{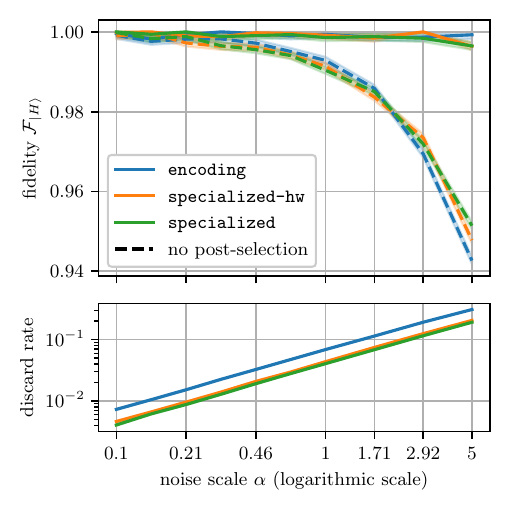}
  \caption{%
    Simulated state fidelity and discard rate (i.e., the probability of a reset and restart) for different $\ket{H}$-state preparation protocols as a function of a multiplicative factor $\alpha$ by which all error rates are scaled up or down in the \texttt{H2-1E} emulator with the \texttt{H2-1} noise model~\cite{moses2023racetrack, 2025system}.
    For reference, $\alpha=1$ corresponds to a probability $p = 1.05\times 10^{-3}$ of an error after a two-qubit gate.
    Here \texttt{encoding} and \texttt{specialized} correspond to the protocols in \cref{fig:prep_h}, while \texttt{specialized-hw} corresponds to the circuit in \cref{fig:prep_h_hw} of \cref{sec:prep_h}, which is used in all hardware experiments in this work.
    Fidelity is computed using \cref{eq:fidelity_h}, for which each expectation value is obtained from 78,000 shots.
    Shaded regions indicate 68\% bootstrap confidence intervals for the fidelity (after maximum-likelihood corrections~\cite{smolin2012efficient}) and standard errors on the mean for the discard rate.
    After post-selection, all $\ket{H}$-state preparation methods have logical fidelities that are statistically indistinguishable from 1, but the specialized circuits have roughly half the discard rate.
  }
  \label{fig:h_state_data}
\end{figure}

We begin by benchmarking the two repeat-until-success protocols in \cref{fig:prep_h} for preparing an $\ket{H}$ state.
We use a round of flag-fault-tolerant syndrome measurement for QED in the method of \cref{fig:prep_h}\sref{fig:prep_h_old}.
We then measure logical $Z$ and $X$ operators, whose expectation values determine the fidelity
\begin{align}
  \mathcal{F}_{\ket{H}}
  = \frac12\sp{1 + \frac{\braket{Z} + \braket{X}}{\sqrt{2}}}.
  \label{eq:fidelity_h}
\end{align}
Each expectation value is computed using the \texttt{H2-1E} emulator with the \texttt{H2-1} noise model~\cite{moses2023racetrack}.
The \texttt{H2-1E} emulator is a classical simulator configured with the calibrated \texttt{H2-1} device noise model and compiler stack~\cite{2025system, yamamoto2025quantum}, which captures the dominant error channels of the device, including: a two-qubit gate fault probability of $1.05\times10^{-3}$ (the largest gate error), a single-qubit gate fault probability of $1.9\times10^{-5}$, measurement errors of order $10^{-3}$, coherent and incoherent memory errors (informed by ion transport schedules), leakage through spontaneous emission, and crosstalk.
Coherent memory errors are mitigated by dynamical decoupling pulses that the \texttt{H2-1} compiler applies on both the emulator and hardware, and errors from ion transport enter through the associated idling and dephasing.
Because the circuits used for our \texttt{H2-1E} benchmarks are small, these simulations use a full statevector simulation.
We refer the reader to Ref.~\cite[Appendix D]{yamamoto2025quantum} and Ref.~\cite{2025system} for a full specification of the noise model.

\cref{fig:h_state_data} shows, for different protocols to prepare a logical $\ket{H}$ state, the fidelity $\mathcal{F}_{\ket{H}}$ and discard rate (that is, the probability of a reset and restart due to a nonzero measurement outcome) as a function of a multiplicative factor $\alpha$ by which all error rates of the \texttt{H2-1} noise model are scaled up or down on the \texttt{H2-1E} emulator.
Varying $\alpha$ probes how each protocol would perform as hardware error rates decrease ($\alpha<1$) or increase ($\alpha>1$), with $\alpha=1$ corresponding to the calibrated device.
Since the expectation values $\braket{Z}$ and $\braket{X}$ are obtained from independent measurements, sampling error (shot noise) may cause their values to correspond to an unphysical state of a qubit.
We therefore apply maximum-likelihood corrections to these expectation values~\cite{smolin2012efficient}, which amounts to rescaling the Bloch vector $\bm v = (\braket{Z},\braket{X})$ to $\bm v / \norm{\bm v}$ whenever the norm $\norm{\bm v}^2 = \braket{Z}^2 + \braket{X}^2 > 1$.

Overall, we find that fidelities of all protocols considered in \cref{fig:h_state_data} are statistically indistinguishable at the sampled 78,000 shots per expectation value.
However, the specialized $\ket{H}$-state circuits (\cref{fig:prep_h}\sref{fig:prep_h_new} and \cref{fig:prep_h_hw}) have roughly half the discard rate of the direct encoding method that was used in previous work~\cite{goto2016minimizing, postler2022demonstration, mayer2024benchmarking} (\cref{fig:prep_h}\sref{fig:prep_h_old}).
We note that this data does not capture additional benefits of the specialized circuits: the reduction in circuit runtime and the reduction of memory error on idling qubits.
The two-qubit gate depth of the specialized circuits is 9, whereas the encoding circuit (with QED via direct measurement of all stabilizers) has a two-qubit gate depth of 25.
Since the runtime of a circuit is approximately linear in two-qubit gate depth, specialized circuits reduce the runtime of a magic-state subroutine, potentially reducing both idling errors for other qubits and the overall runtime of the quantum computation, which is typically bottlenecked by the runtime of such subroutines.

\subsubsection{$T$ gates}

\begin{figure}
  \centering
  \includegraphics{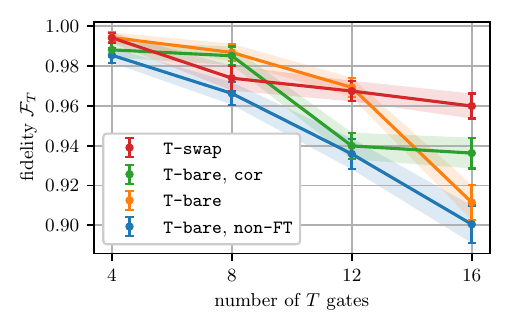}
  \caption{%
    Experimental Ramsey fidelities obtained from \texttt{H2-1} using $T$ gates that are applied with the gadgets in \cref{fig:t_gate_and_qec}\sref{fig:t_gate_bare} (\texttt{T-bare}) and \cref{fig:t_gate_and_qec}\sref{fig:t_gate_swap} (\texttt{T-swap}).
    In the case of \texttt{non-FT}, logical $\ket{T}$ states are prepared non-fault-tolerantly using the encoding subcircuit of \cref{fig:prep_h}\sref{fig:prep_h_old}, without a logical $H$ measurement or QED.
    In all other cases, logical $\ket{T}$ states are prepared fault tolerantly using the circuit in \cref{fig:prep_h_hw} of \cref{sec:prep_h} (an older variant of \cref{fig:prep_h}\sref{fig:prep_h_new}).
    Experiments labeled \texttt{cor} apply active bit-flip corrections after every logical $T$ gate based on a measured syndrome.
    Every data point is obtained from 1000 shots.
    Error bars indicate $68\%$ Wilson confidence intervals.
    Interpolating lines and shaded regions are provided as a guide to the eye.
    Fault-tolerant gadgets outperform the non-fault-tolerant gadget used in prior work, and the \texttt{T-swap} generally performs best.
  }
  \label{fig:t_gate_data}
\end{figure}

\begin{figure}
  \centering
  \includegraphics{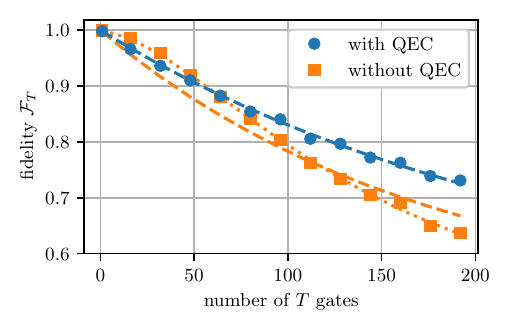}
  \caption{%
    Simulated Ramsey fidelities as a function of the number of $T$ gates with and without a QEC cycle after every $T$ gate.
    Dashed lines show fits to a single-parameter exponential decay model (\cref{sec:ramsey}), while the dotted line shows a fit to the a two-parameter decay model that accounts for the effect of phase error accumulation throughout a Ramsey sequence without intermediate QEC cycles (\cref{sec:non_exponential}).
    Adding QEC cycles recovers exponential decay of fidelity, but at the cost of decreased Ramsey fidelity for $\lesssim64$ $T$ gates, beyond which QEC improves the Ramsey fidelity.
    The results here were obtained from stabilizer simulations of a Clifford proxy for the Ramsey sequence, using the \texttt{Selene} simulator of \texttt{Guppy} programs with two-qubit depolarizing gate rates of $2\times 10^{-3}$, measurement error rates of $10^{-3}$, an RUS limit of 3, and the \texttt{T-swap} protocol.
  }
  \label{fig:t_gate_sim}
\end{figure}

We now benchmark variants of the Steane-encoded $T$ gate on \texttt{H2-1}.
To this end, we run a Ramsey protocol wherein we prepare a logical $\ket{+}$ state, apply $n_T$ sequential $T$ gates (where $n_T$ is a multiple of four), and perform a logical $X$-basis measurement.
The Ramsey fidelity $\mathcal{F}_T$ of the prepared state with respect to the target state is then the fraction of times that the state is found to be $T^{n_T}\ket{+}\in\set{\ket{+},\ket{-}}$ upon measurement.
Following the procedure proposed in Ref.~\cite{piveteau2021error} (and used in Ref.~\cite{mayer2024benchmarking}), our $T$-gate benchmarks twirl $\ket{T}$ states by appending $SX$ to every $\ket{T}$-state preparation gadget with probability $1/2$ (determined online with a pseudo-random number generator).

\cref{fig:t_gate_data} shows Ramsey fidelities obtained on \texttt{H2-1} using logical $T$ gates that are applied with both \texttt{T-bare} and \texttt{T-swap} gadgets in \cref{fig:t_gate_and_qec}, for which logical $\ket{T} = HS\ket{H}$ states are prepared using the specialized circuit for $\ket{H}$ in \cref{fig:prep_h_hw} of \cref{sec:prep_h} (an older version of the circuit in \cref{fig:prep_h}\sref{fig:prep_h_new}).
For reference, we also plot Ramsey fidelities obtained in Ref.~\cite{mayer2024benchmarking} by preparing $\ket{T}$ states non-fault-tolerantly via the encoding subcircuit of \cref{fig:prep_h}\sref{fig:prep_h_old}, without a logical measurement of $H$ or QED, as well as fidelities that we obtain identically on \texttt{H2-2}.

Due to the limited support for conditional ion transport on \texttt{H2-1}, the results in \cref{fig:t_gate_data} are post-selected on the successful preparation of logical $\ket{H}$ states, for which all ancilla measurement outcomes are $0$.
This post-selection amounts to neglecting idling errors that would otherwise accumulate whenever the $\ket{H}$-state preparation protocol fails and restarts.
We allow for the conditional reset and restart of $\ket{H}$-state preparation protocols in our application-level benchmarking on \texttt{Helios} in \cref{sec:results}.

The contribution of $T$ gate errors to the Ramsey fidelity $\mathcal{F}_T(n_T)$ grows with the sequence length $n_T$, while the contribution of state preparation and measurement error is independent of $n_T$.
Ramsey sequences can therefore be used to estimate the infidelity of a single $T$ gate as $1 - \mathcal{F}_T(n_T)^{1/n_T}$.
The Ramsey fidelity of $0.960(6)$ for \texttt{T-swap} at 16 $T$ gates thereby corresponds to a $T$ gate infidelity of $2.6(4)\times10^{-3}$.
For comparison, an identical calculation from prior work \cite[Fig.~5, T method 2]{mayer2024benchmarking} yields a logical $T$ gate infidelity of $14\times10^{-3}$.

The $T$ gate infidelity can also be computed by fitting experimental data on $\mathcal{F}_T(n_T)$ to a single-parameter logical dephasing model, \texttt{DepModel} (see \cref{sec:ramsey} for details), which yields a $T$ gate infidelity of $2.7(4)\times10^{-3}$ for the \texttt{T-swap} gate.
However, \texttt{DepModel} predicts an exponential decay of the Ramsey fidelity with respect to the number of $T$ gates, which is a poor fit to the experimental data in \cref{fig:t_gate_data}.
The observed non-exponential behavior reflects a violation of the assumption in \texttt{DepModel} that logical $T$ gates cause independent and identically distributed (IID) logical errors (though we note that other effects, such as coherent logical errors and non-Markovian noise \cite{ceasura2022nonexponential}, may also contribute to a non-exponential decay of fidelity).
In contrast, the Ramsey experiments in \cref{fig:t_gate_data} and Ref.~\cite{mayer2024benchmarking} allow data-qubit (physical) $Z$ errors to accumulate throughout the entire circuit.
The question of whether a logical error has occurred therefore depends on how pairs of single-qubit $Z$ errors are correlated \emph{across} different $T$ gate gadgets.
As a consequence, the overall probability of a logical error acquires an additional term that is quadratic in $n_T$, leading to a Gaussian decay of the Ramsey fidelity (see \cref{sec:non_exponential}).

The IID assumption of the \texttt{DepModel} can be enforced by the use of extended rectangles~\cite{aliferis2005quantum}, which amounts to inserting a QEC cycle between every pair of $T$ gates.
However, adding such a high density of QEC cycles may decrease the Ramsey fidelity due to increased opportunities for gate, measurement, and idling errors.
We illustrate this effect in \cref{fig:t_gate_sim}, which shows the decay of fidelity in Ramsey protocols with and without QEC cycles, performed using the \texttt{Selene} simulator of \texttt{Guppy} programs~\cite{koch2024guppy, koch2025imperative}.
Notably, adding a QEC cycle between every pair of logical $T$ gates recovers exponential decay behavior, but decreases the overall Ramsey fidelity when the total number of $T$ gates is $\lesssim64$.
This crossover follows directly from the Gaussian decay model of \cref{sec:non_exponential}: a QEC cycle after every $T$ gate corrects accumulating $Z$ errors, but incurs a fixed per-gate cost set by the infidelity of a QEC cycle.
QEC therefore helps only for long sequences in which the infidelity due to phase error accumulation exceeds the infidelity of a QEC cycle.
The subtle technical challenge of quantifying---or even defining---the fidelity of an isolated logical $T$ gate motivates a shift in focus beyond component benchmarks towards application-level hardware benchmarking.
Nonetheless, we remark that fitting the \texttt{T-swap} data in \cref{fig:t_gate_data} to the non-exponential decay model of \cref{sec:non_exponential} yields the same infidelity as the \texttt{DepModel}, of $1 - \mathcal{F}_T(1) = 2.7(4)\times10^{-3}$.

\subsubsection{Quantum error correction}

\begin{figure}
  \centering
  \includegraphics{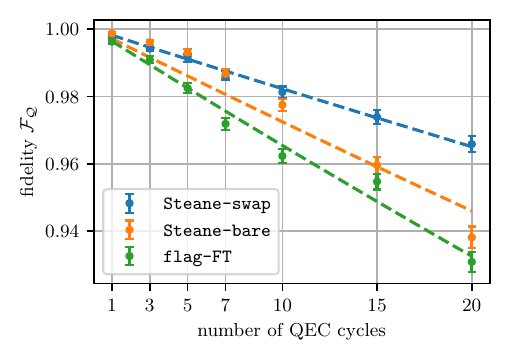}
  \caption{%
    Simulated entanglement fidelity as a function of the number of QEC cycles performed using different QEC gadgets:
    flag-fault-tolerant QEC (\texttt{flag-FT}) and two types of Steane-type QEC (\texttt{Steane-bare} and \texttt{Steane-swap}).
    For the Steane-type methods, a single QEC cycle consists of two rounds of error correction, one each for $X$-type and $Z$-type errors.
    QEC cycles were simulated on the \texttt{H2-1E} emulator with $2{,}000$ shots for each of the six initial states used to obtain a single entanglement fidelity.
    Dashed lines show a fit to a single-parameter depolarizing noise model.
    Error bars indicate standard errors on the mean.
  }
  \label{fig:qec_data}
\end{figure}

We now benchmark different strategies for performing QEC.
Letting $\mathcal{Q}$ denote the quantum channel that one (noisy) QEC cycle applies to a logical qubit, the channel fidelity of $n$ QEC cycles with respect to the identity channel $\mathcal{I}$ can be mathematically obtained by applying $\mathcal{Q}^n$ to one qubit of a maximally entangled Bell pair $\ket\phi = \frac1{\sqrt{2}}(\ket{00} + \ket{11})$ and measuring the fidelity of the resulting state with the initial Bell state, altogether arriving at the entanglement fidelity~\cite{nielsen1996entanglement}
\begin{align}
  \mathcal{F}_{\mathcal{Q}}(n)
  = \braket{\phi|(\mathcal{I}\otimes\mathcal{Q}^n)(\phi)|\phi},
\end{align}
where $\phi = \op{\phi}$.
Letting $\Pi_{P,b}$ denote the projector onto the eigenstate of the Pauli operator $P\in\set{X,Y,Z}$ with eigenvalue $(-1)^b$, such that for example $\Pi_{Z,0} = \op{0}$, $\Pi_{Z,1} = \op{1}$, and $\Pi_{X,0} = \op{+}$, the entanglement fidelity can equivalently be written as (see \cref{sec:channel_fidelity})
\begin{align}
  \mathcal{F}_{\mathcal{Q}}(n)
  = \frac14 \sum_{\substack{P\in\set{X,Y,Z}\\b\in\set{0,1}}}
  \Tr\sp{\mathcal{Q}^n(\Pi_{P,b}) \Pi_{P,b}} - \frac12.
  \label{eq:fidelity_qec}
\end{align}
In words, $\Tr\sp{\mathcal{Q}^n(\Pi_{P,b})\Pi_{P,b}}$ is the survival probability of an initial state $\Pi_{P,b}$ after $n$ applications of the channel $\mathcal{Q}$; that is, the probability that the state is still found to be $\Pi_{P,b}$ upon measurement in the $P$ basis.
The QEC fidelity $\mathcal{F}_{\mathcal{Q}}(n)$ can therefore be computed by preparing each of six Pauli eigenstates $\Pi_{P,b}$, applying repeated rounds of QEC, measuring survival probabilities $\Tr\sp{\mathcal{Q}^n(\Pi_{P,b})\Pi_{P,b}}$, and combining these probabilities according to \cref{eq:fidelity_qec}.

\cref{fig:qec_data} shows entanglement fidelities obtained from the \texttt{H2-1E} hardware emulator for each of the flag-fault-tolerant (\texttt{flag-FT}), \texttt{Steane-bare}, and \texttt{Steane-swap} QEC protocols in \cref{sec:qec}.
We note that the limited support for conditional ion transport in the compiler for \texttt{H2-1} (and \texttt{H2-1E}) disadvantages the flag-fault-tolerant QEC strategy, for which ions are transported (albeit without gates or ion cooling) for a secondary round of QED even if the first round of flagged stabilizer measurements yields a trivial syndrome.
Altogether, we conclude from \cref{fig:qec_data} that the \texttt{Steane-swap} protocol outperforms \texttt{Steane-bare} and \texttt{flag-FT}, though future improvements to the \texttt{H2-1} compiler may improve the fidelity of \texttt{flag-FT}.

A single-parameter logical depolarizing model in which each QEC cycle $\mathcal{Q}$ has probability $p$ to replace the logical state of a qubit with the maximally mixed state predicts that these fidelities should decay as (see \cref{sec:channel_fidelity})
\begin{align}
  \mathcal{F}_{\mathcal{Q}}(n) = \frac14 + \frac34 (1-p)^n.
  \label{eq:qec_decay_model}
\end{align}
Fitting the simulation data in \cref{fig:qec_data} to this logical depolarizing noise model, we find entanglement infidelities (defined by $\epsilon = 1 - \mathcal{F}_{\mathcal{Q}}(1) = \frac34 p$) of $3.53(12)\times10^{-3}$ for the \texttt{flag-FT} protocol, and $1.79(3)\times10^{-3}$ for the \texttt{Steane-swap} protocol, where uncertainties in the final digits reflect the standard error from the fit.
The principled inference of an entanglement infidelity for the \texttt{Steane-bare} protocol is obstructed by its non-exponential decay in \cref{fig:qec_data}.
We attribute this non-exponential decay to the accumulation of leakage errors, analogous to the accumulation of $Z$ errors in a Ramsey sequence (see \cref{fig:t_gate_sim}): because the \texttt{Steane-bare} gadget retains the logical state on the same physical data qubits from cycle to cycle, leakage accumulates without being removed.
The \texttt{Steane-swap} gadget, by contrast, teleports the logical state onto fresh qubits every cycle and thereby mitigates leakage (see \cref{fig:t_gate_and_qec}\sref{fig:qec_steane_swap}), consistent with its clean exponential decay.
The \texttt{flag-FT} experiment also suffers from the accumulation of leakage errors, but additionally exhibits irregular and non-Markovian behavior due to the intrinsically dynamic nature the \texttt{flag-FT} QEC cycle (which uses measurement outcomes to make real-time choices about executing subroutines).
These observations do not impact our primary conclusion that \texttt{Steane-swap} is the preferred method for performing active QEC.

\subsection{Application benchmarking}
\label{sec:results_applications}

Having benchmarked the components that are used to run error-corrected quantum algorithms, we now consider application-level benchmarks.
To this end, we construct various QAOA and HHL circuits, compile these circuits into a Clifford+$T$ gate set (as discussed in \cref{sec:compilation} and \cref{sec:hhl_details}), and expand each compiled logical circuit into a physical circuit using Steane-code gadgets.
Specifically, we use the specialized $\ket{H}$-state preparation circuit in \cref{fig:prep_h_hw} of \cref{sec:prep_h} (an older version of the circuit in \cref{fig:prep_h}\sref{fig:prep_h_new}), the \texttt{T-swap} variant of the $T$ gate in \cref{fig:t_gate_and_qec}\sref{fig:t_gate_swap}, and (when applicable) the \texttt{Steane-swap} variant of a QEC cycle in \cref{fig:t_gate_and_qec}\sref{fig:qec_steane_swap}.
In some cases, we additionally consider varying the RUS limit.

We execute logical circuits on \texttt{H2-2} and \texttt{Helios}.
While both of these devices support real-time classical feedback, such as conditioning a quantum gate on the outcome of a preceding measurement, \texttt{H2-2} programs must nonetheless be compiled into an intermediate representation of serialized instructions prior to execution.
In contrast, \texttt{Helios} supports a real-time engine that consumes \texttt{Guppy}~\cite{koch2024guppy, koch2025imperative} programs, which allow for nonlinear control flow, such as conditional branching and \texttt{FOR} loops.
Moreover, the software stack for \texttt{Helios} allows for real-time dynamic allocation of physical qubits, whereas executing circuits on \texttt{H2-2} requires allocating all physical qubits at compile-time.
In particular, to execute our circuits on \texttt{H2-2} (which has 56 physical qubits) we allocate a fixed choice of seven qubits per Steane code block (including one ancilla code block for $T$ gate teleportation and Steane QEC), and additionally six physical qubits as ancillas that are shared by the Steane code blocks.

\subsubsection{QAOA}
\label{sec:qaoa_experiments}

\begin{figure*}
  \centering
  \includegraphics{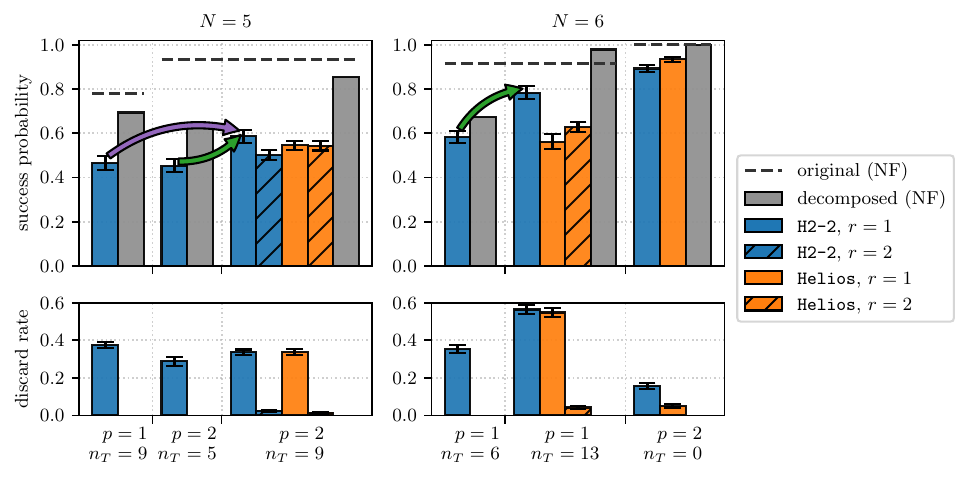}
  \caption{%
    Probability of measuring the optimal solution to the LABS problem for $N=5,6$ bits.
    Data is grouped by the number of QAOA layers, $p$, and the total number of $T$ gates in a compiled circuit, $n_T$.
    Dashed lines indicate the success probability for exact, noise-free (NF) QAOA circuits with optimal parameters at a given QAOA depth $p$.
    Grey bars show the noise-free success probability after decomposition to a Clifford+$T$ gate set (which exceeds the undecomposed circuit at $(N,p,n_T)=(6,1,13)$ due to a finite-size effect; note that the decomposed circuit is not limited to the QAOA ansatz structure).
    Orange and blue bars show, respectively, the success probability when running on \texttt{H2-2} and \texttt{Helios}.
    Solid and hatched bars correspond to RUS limits of $r=1$ and $r=2$, respectively.
    Bottom panels show the fraction of hardware runs that are discarded because a subroutine exceeded the RUS limit.
    Error bars indicate standard errors on the mean.
    Increasing $p$ or decomposing to a larger number of $T$ gates increases the success probability on \texttt{H2-2} (purple and green arrows), indicating that the benefits of finer decomposition outweigh the increased errors from running a larger logical circuit.
    Increasing the RUS limit $r$ from $1$ to $2$ slightly decreases the success probability on \texttt{H2-2}, but does not significantly affect the success probability on \texttt{Helios}.
    However, increasing the RUS limit to $r=2$ reduces the program discard rate to nearly zero.
  }
  \label{fig:labs_n5-6}
\end{figure*}

We first consider QAOA circuits for the LABS problem with binary sequences of lengths $N=5,6$ (see \cref{sec:labs}).
\cref{fig:labs_n5-6} shows the measured QAOA success probability (that is, the probability of measuring the optimal solution to a LABS instance) with various choices of the device, QAOA depth $p$, $T$ gate count $n_T$, and RUS limit $r$ (see \cref{tab:labs_5} and \cref{tab:labs_6} of \cref{sec:data_tables} for additional experimental data).
While the number of qubits and $T$ gates in these logical circuits are modest, we emphasize that the corresponding physical circuits are considerably larger, involving up to 55 qubits and several hundred two-qubit gates (see \Cref{tab:qaoa_po,tab:labs_6,tab:labs_hero_run,tab:labs_5}).
Remarkably, for both the $N=5$ and $N=6$ LABS instances we observe that increasing the logical $T$ gate budget improves QAOA performance at depth $p=2$ (green arrows in \cref{fig:labs_n5-6}), despite a considerable increase in physical circuit complexity and two-qubit gate count ($399\to551$ for $N=5$ and $420\to686$ for $N=6$ when $r=1$).
We also observe on \texttt{H2-2} that for $N=5$ the QAOA success probability can be increased by increasing the QAOA depth $p:1\to 2$ at a fixed logical $T$ gate budget of $n_T=9$.
Increasing QAOA depth similarly improves the performance for $N=6$; however, in this case the $p=2$ logical circuit is Clifford, making the physical circuit for $p=2$ slightly smaller than that for $p=1$.

Where data is available, we observe that increasing the RUS limit $r$ slightly decreases the QAOA success probability on \texttt{H2-2}, and has no meaningful impact on the success probability on \texttt{Helios}.
However, increasing the RUS limit to $r=2$ drastically reduces the discard rate of hardware executions, or the probability with which a program must terminate and restart due to a corrupted subroutine, to nearly zero.
This improvement highlights the importance of dynamic repeat-until-success protocols for making fault-tolerant quantum computation scalable.
We remark that the benefit of repeat-until-success on \texttt{Helios} is observed despite the lack of dynamical decoupling in the \texttt{Helios} compiler at the time of these experiments.
The addition of dynamical decoupling, as well as other compiler improvements, can therefore be expected to yield substantial improvements in future \texttt{Helios} performance.

\begin{figure}
  \centering
  \includegraphics{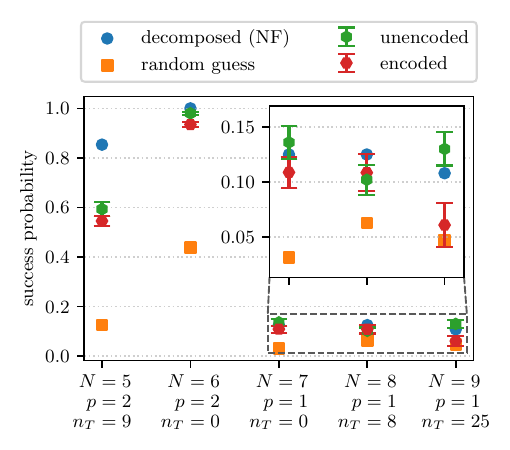}
  \caption{%
    Success probability for QAOA circuits with QAOA depth $p$ on LABS instances with varying problem size $N$.
    ``Decomposed (NF)'' corresponds to a noise-free circuit that has been decomposed to a Clifford+$T$ gate set (with $n_T$ $T$ gates in the decomposed circuit).
    ``Unencoded'' and ``encoded'' correspond, respectively, to unencoded (physical) and encoded (logical, RUS limit $r=1$) execution of the same circuit on \texttt{Helios}.
    ``Random guess'' shows the probability of success for a randomly chosen bitstring.
    Error bars indicate standard errors on the mean.
    Decomposed and encoded circuit data for $N=5$ and $N=6$ is the same as in \cref{fig:labs_n5-6}.
    The encoded circuit performs comparably to the unencoded one up through $N=8$, but performance deteriorates at $N=9$ due to the large number of $T$ gates.
  }
  \label{fig:labs_hero_run}
\end{figure}

To assess the scalability of our fault-tolerant QAOA implementation on \texttt{Helios}, we execute logical circuits for LABS problem instances up to $N=9$ variables and a portfolio optimization problem instance with $N=12$ variables.
For LABS instances with $N>6$, the pre-optimized QAOA parameters in \texttt{qokit} yield circuits that, when decomposed into a Clifford+$T$ gate set, exceed the $T$-gate budget that is executable on \texttt{Helios} within practical coherence times.
To identify parameter settings that produce executable circuits while maintaining the performance of the algorithm, we perform a grid search over the two QAOA parameters at depth $p=1$, and select the best-performing parameters that yield compiled circuits with fewer than 30 $T$ gates.
The results from running these LABS experiments are summarized in \cref{fig:labs_hero_run}, and additional data is provided in \cref{tab:labs_hero_run} of \cref{sec:data_tables}.
For LABS instances of size up to $N=8$, we find that the logically encoded circuits yield a probability of success that is significantly greater than a random guess, and is moreover comparable to the probability of success obtained by unencoded circuit execution.
These results suggest the possibility that future hardware upgrades, software-level improvements (such as the addition of dynamical decoupling), and circuit-level decoding may push the performance of error-corrected algorithms beyond break-even, i.e., with logically encoded circuits outperforming unencoded circuits.

\begin{figure}
  \centering
  \includegraphics{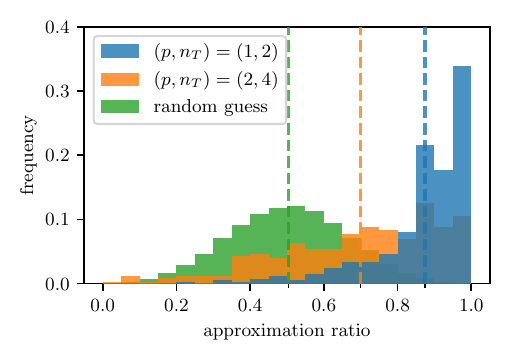}
  \caption{%
    Histogram of approximation ratios for a portfolio optimization problem instance with $N=12$ binary variables, obtained by running logical QAOA on \texttt{Helios} with an RUS limit of $r=1$.
    Vertical axis values indicate the fraction of outcomes that fall within each bin (of width $0.05$).
    The distribution of approximation ratios for randomly chosen bitstrings is provided for reference.
    Dashed vertical lines indicate distribution means.
  }
  \label{fig:qaoa_po}
\end{figure}

Finally, we scale QAOA for portfolio optimization to 12 logical qubits, using 97 of the 98 qubits available on the \texttt{Quantinuum Helios} processor.
To identify a portfolio optimization instance within the current capabilities of \texttt{Helios}, we generate multiple $N=12$ problem instances by randomly sampling mean return vectors and covariance matrices from Yahoo Finance data using \texttt{qokit}~\cite{2025qokit}.
For each problem instance, we numerically optimize the QAOA parameters and compile the resulting circuits using the procedure described in \cref{sec:compilation}.
We then select an instance whose compiled logical circuits at QAOA depths $p=1,2$ have small numbers of $T$ gates.
\cref{fig:qaoa_po} shows a histogram of approximation ratios for the logical bitstrings sampled from these circuits on \texttt{Helios}; additional data is provided in \cref{tab:qaoa_po} of \cref{sec:data_tables}.
We remark that these circuits use $7\times12=84$ physical qubits to encode the $N=12$ problem variables, 7 qubits for a logical ancilla qubit, and additionally 6 physical ancilla qubits, for a total of 97 physical qubits.
The compiled $p=1$ circuit has 1132 physical two-qubit gates, and the $p=2$ circuit has 2132 physical two-qubit gates.
Although the $p=2$ circuit performs worse than that with $p=1$ due to its larger physical gate count, both circuits outperform random guessing ($p=0$).
While this portfolio optimization problem is highly simplified and the circuits are produced in a non-scalable way, our results show that fault-tolerant and error-corrected execution of algorithmic circuits with $12$ logical qubits can achieve high fidelity.

\subsubsection{HHL}
\label{sec:hhl_experiments}

For our final application benchmark, we run the HHL circuit in \cref{fig:circuit_hhl}\sref{fig:hhl_poisson} on \texttt{Helios}.
After decomposing into a Clifford+$T$ gate set (see \cref{sec:hhl_details} for details) and post-selecting on a measurement of $0$ on the last ancilla qubit, this circuit prepares the state
\begin{align}
  \ket\chi
  = \left(\frac{\ket{0} + i e^{i\eta} \ket{1}}{\sqrt{2}}\right)
  \otimes \left(\frac{\ket{0} - \ket{1}}{\sqrt{2}}\right),
  \label{eq:state_hhl}
\end{align}
where $\eta = \arccos\sqrt{8/9} \approx 0.11 \pi$.
The undecomposed circuit prepares the same state with $\eta=0$.
The fidelity $\mathcal{F}_{\ket\chi}(\rho) = \abs{\braket{\chi|\rho|\chi}}^2$ of a (possibly mixed) state $\rho$ with respect to $\ket\chi$ can be recovered from measurements of the first qubit in the $X$ and $Y$ bases, and the second qubit in the $X$ basis.
In the absence of errors, the \texttt{QPE} qubits are guaranteed to end in logical $\ket{0}$ states.
This fact can be leveraged to perform algorithm-level error detection by measuring the \texttt{QPE} qubits and discarding shots in which the \texttt{QPE} measurement outcomes are nonzero.

\begin{figure*}
  \centering
  \includegraphics{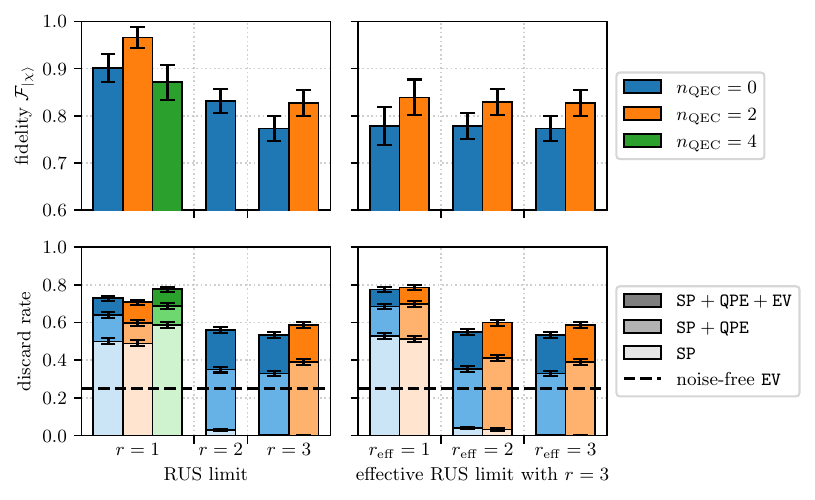}
  \caption{%
    Target state fidelities and discard rates of logical HHL circuits executed on \texttt{Helios} with varying RUS limits $r$ (left panels), or when $r=3$ and every RUS subroutine succeeded in at most $r_{\mathrm{eff}}$ attempts (right panels).
    Color indicates the number of QEC cycles, $n_{\mathrm{QEC}}$, inserted into the HHL circuit (see \cref{fig:hhl_compiled} of \cref{sec:hhl_details} for the compiled circuit and QEC cycle locations).
    Shade indicates post-selection criteria: trivial flags at the end of all state preparation routines (\texttt{SP}), logical measurement outcomes of $0$ on the \texttt{QPE} qubits, and a logical measurement outcome of $0$ on the \texttt{EV} qubit.
    The dashed line in the bottom panels indicates the \texttt{EV} discard rate of the noise-free compiled HHL circuit.
    Error bars indicate $68\%$ confidence intervals obtained by bootstrapping (top panels) or standard errors on the mean (bottom panels).
    Adding two QEC cycles ($n_{\mathrm{QEC}}=2$) improves logical circuit fidelity, but additional QEC cycles can introduce more errors than they mitigate ($n_{\mathrm{QEC}}=4$), reflecting that the performance of Steane-encoded algorithms is not yet at break-even.
    Increasing the RUS limit drops the failure rate of state preparation subroutines to nearly zero, but at a moderate cost to the logical fidelity of the HHL circuit.
    However, this fidelity penalty depends on the compile-time RUS limit $r$ (top left panel), not the maximal number of attempts actually used at runtime, $r_{\mathrm{eff}}$ (top right panel), indicating that the fidelity penalty is a result of compile-time behavior rather than runtime or hardware effects.
  }
  \label{fig:hhl_rus}
\end{figure*}

\cref{fig:hhl_rus} shows the fidelities and discard rates of logical HHL circuits executed on \texttt{Helios} with different RUS limits $r$ and numbers of QEC cycles, $n_{\mathrm{QEC}}$ (see \cref{fig:hhl_compiled} of \cref{sec:hhl_details} for the compiled logical circuit and the locations of QEC cycles, as well as \cref{tab:hhl_data} of \cref{sec:data_tables} for additional experimental data).
The physical circuits for these experiments use 48 physical qubits and, at an RUS limit of $r=1$, have two-qubit gate counts of $(677, 749, 821)$ for $n_{\mathrm{QEC}}=(0,2,4)$.
We find that the addition of two QEC cycles in the middle of the circuit improves the fidelity of the measured logical state with respect to the expected state in \cref{eq:state_hhl}, which can be as high as $0.97\pm0.02$.
However, additional QEC cycles can degrade the performance of the algorithm, reflecting that the Steane code is near---but not yet below---the break-even point for these algorithms.
Similarly to our QAOA benchmarking, we find that increasing the RUS limit to $r=2$ decreases the probability of failure for state preparation subroutines to nearly zero (see \texttt{SP} discard rates in \cref{fig:hhl_rus}).
Unlike the case of QAOA, however, this benefit comes at the cost of reduced logical circuit fidelity for the HHL circuit (see top left panel of \cref{fig:hhl_rus}).

To diagnose the reason for a fidelity penalty with higher RUS limits, we examine the data from programs with an RUS limit of $r=3$ more closely, and consider the shots in which all RUS subroutines succeed in at most $r_{\mathrm{eff}}\in\set{1,2,3}$ attempts.
This analysis acts as an experimental control to isolate effects that occur in hardware from those that occur at the software level.
In principle, if a program has an RUS limit of $r$ but all RUS subroutines succeed in at most $r_{\mathrm{eff}} \le r$ attempts, the results on hardware should be identical to those of a program with an RUS limit of $r_{\mathrm{eff}}$.
This observation lets us diagnose the source of the fidelity penalty.
If the penalty came from actually executing more state-preparation attempts, it would depend on the maximal number attempts used at runtime, $r_{\mathrm{eff}}$.
Instead, we find that the fidelity depends on the compile-time RUS limit $r$ (top left panel of \cref{fig:hhl_rus}) but is essentially independent of $r_{\mathrm{eff}}$ (top right panel).
The fidelity penalty is therefore attributed to compile-time behavior rather than runtime or hardware effects.
We attribute this effect to ion-transport scheduling in the QCCD architecture: each repeat-until-success branch introduces conditional branching logic that enforces constraints on the transport scheduler.
Compiling a dynamic program with a larger RUS limit may therefore yield a less efficient transport schedule with more ion movement and idling throughout the circuit, even if the extra repeat attempts are never used.
Importantly, the compiled circuits remain logically correct, reproducing the intended circuit in noise-free simulation; the RUS limit affects only scheduling efficiency, and hence idling error.
Because this is a scheduling effect and not a miscompilation, it does not undermine our other results, which always compare circuits compiled in the same way; here it appears only as the modest, explicitly reported fidelity cost of raising the RUS limit.
This observation highlights the importance of application-level benchmarking for uncovering behaviors of a quantum computer that isolated component benchmarking may not reveal.

Finally, we observe in \cref{fig:hhl_rus} that the probability of obtaining nonzero logical measurement outcomes on the \texttt{QPE} qubits is consistently high ($>0.3$), implying a large logical error rate, even in the case of $r\ge2$ when all state preparation subroutines succeed with high probability (indicated by a low \texttt{SP} discard rate).
We remark that the use of fault-tolerant state preparation subroutines, which account for the vast number of physical gates in the HHL circuit, suppresses the physical gate errors of $p\sim10^{-3}$ on \texttt{Helios} to logical error rates of $\mathcal{O}(p^2)$.
The large number of physical gates in the HHL circuit therefore cannot account for the observed logical error rates.
It is unclear whether idling errors can account for the observed logical error rates, though upgrades to the \texttt{Helios} compiler---such as the addition of dynamical decoupling---should mitigate idling errors and provide clarity on their impact.
These logical error rates also reflect the need for either extended rectangles (by, for example, adding a QEC cycle on every logical qubit between every pair of gates)~\cite{aliferis2005quantum} or correlated decoding (in the spirit of algorithmic fault tolerance~\cite{cain2024correlated, zhou2025lowoverhead}) to control the propagation of physical errors throughout the circuit.
Without such techniques, physical errors can propagate to logical errors even when using transversal gates and fault-tolerant gadgets (see \cref{sec:decoding} for a toy example).
In practice, the construction of extended rectangles incurs substantial quantum computing overheads, whereas correlated decoding can be performed at the software level with modest classical computing overheads (see for example Ref.~\cite{cain2025fast}).
We therefore anticipate that future efforts to fault-tolerantly execute error-corrected quantum algorithms on trapped-ion quantum architectures will require dynamical decoupling and more advanced decoding techniques to further suppress logical error rates.

\section{Discussion}
\label{sec:discussion}

We initiate the application-level benchmarking of quantum processors at the logical level with both error correction and fault tolerance, which was previously obstructed by small qubit counts and low fidelities.
Our results show that existing trapped-ion quantum computers, namely, \texttt{Quantinuum System Model H2} and \texttt{Quantinuum Helios}, operate near break-even for complex algorithmic circuits that are logically encoded into the Steane code, altogether involving dozens of physical qubits and hundreds of gates.
Furthermore, our experiments suggest that modest improvements to fault-tolerant circuit gadgets, dynamic program compilation, correlated decoding, or the choice of QEC code used may be sufficient to achieve the algorithmic break-even point~\cite{hao2025compilation} at which fault-tolerantly encoded quantum programs outperform their unencoded counterparts.

QAOA has long been used as a benchmark for physical quantum processors.
In the absence of noise, increasing QAOA depth can only improve performance.
However, on noisy devices increasing depth causes a corresponding growth of errors, leading to the existence of an optimal QAOA depth $p_\star$ beyond which performance no longer improves.
Implementing $N$-qubit QAOA with a large $N\cdot p_\star$ was the central goal of the DARPA ONISQ program~\cite{DARPA_ONISQ}.
Multiple independent benchmarking studies have used this quantity to characterize the capabilities of physical~\cite{shaydulin2023qaoa, pelofske2024scaling, pelofske2023highround, montanez-barrera2025evaluating, lubinski2024optimization} or partially fault-tolerant~\cite{he2025performance, jin2025iceberg} devices.
In this work, we extend this approach to benchmarking fault-tolerant error-corrected processors, achieving $N\cdot p_\star = 6\cdot 2 = 12$.

Targeting application-level benchmarks yields several important insights.
First, it provides an overview and synthesis of all concepts and components that must be integrated for fault-tolerant and error-corrected quantum program execution.
As devices continue to scale up, it is important to continue identifying and benchmarking alternative methods for each necessary function, such as active QEC or magic state injection.
Second, our application benchmarks demonstrate key capabilities, most notably the ability to suppress errors through active QEC and to reduce discard rates through dynamic (RUS) subroutines, at minimal or no cost to algorithm performance.
Finally, these benchmarks reveal gaps and challenges that need to be addressed in future work.
The compiled HHL circuit, for example, requires 14 magic states and 28 two-qubit gates, but its logical error rate is substantially worse than the component logical error rates of $\sim 10^{-3}$ would suggest (e.g., a rough estimate of $1 - 0.999^{14}\approx 0.01$ obtained by counting logical $T$ gates, or even $1 - 0.999^{14+28} \approx 0.04$ by counting logical $T$ and CNOT gates).
Instead, our application benchmarks reveal the importance of factors such as the compilation of dynamic programs, the fault-tolerant \emph{composition} of fault-tolerant gadgets (e.g., via extended rectangles), and correlated or circuit-level decoding.
The shift in focus beyond components and error rates as bottlenecks of quantum algorithm performance simultaneously speaks to the maturity of the field and signals the advent of the early-fault-tolerant era for quantum computing.

\section{Acknowledgments}

We thank Matthew DeCross and Kentaro Yamamoto for helpful feedback on a draft of this manuscript.
We acknowledge the entire Quantinuum team for innumerable contributions towards the design and operation of the \texttt{H2} and \texttt{Helios} quantum computers, and we acknowledge Honeywell for fabricating the ion trap used in this experiment.
We thank Rob Otter for executive support of this work and invaluable feedback on this project.
We thank the technical staff at JPMorganChase's Global Technology Applied Research for their support.

\section{Data Availability Statement}

Most data used in this work are provided in \cref{sec:data_tables}.
Additional data are available from the authors upon reasonable request.

\bibliographystyle{apsrev4-2-author-truncate}
\bibliography{main.bib}

\appendix
\onecolumngrid

\section{Logical errors and correlated decoding}
\label{sec:decoding}

Here we consider a toy example of a logical circuit in which physical error rates of $p$ lead to logical error rates of $\mathcal{O}(p)$, even when encoding into the Steane code using fault-tolerant gadgets:
\begin{align}
  \includegraphics[valign=c]{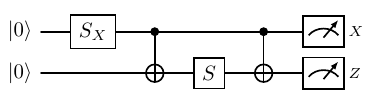}.
  \label{eq:decoding_example}
\end{align}
The two-qubit state in this circuit transforms as
\begin{align}
  \ket{00}
  \xrightarrow{S_X} \ket{00} - i\ket{10}
  \xrightarrow{CX} \ket{00} - i\ket{11}
  \xrightarrow{S} \ket{00} + \ket{11}
  \xrightarrow{CX} \ket{00} + \ket{10}
  = \ket{+} \otimes \ket{0}.
\end{align}
In the absence of errors, both final measurement outcomes should therefore be $0$.
However, for CSS codes such as the Steane code, fault tolerance allows state preparation subroutines to end with weight-2 errors of mixed ($X$ and $Z$) type, such as $X_0 Z_1$.
In turn, $S$-type gates can change the type of one of these errors, for example with an $S_X$ gate converting $X_0 Z_1 \to X_0 X_1 \cdot Z_1$, and logical CNOT gates can propagate such errors between logical qubits in such a way that naive, uncorrelated decoding leads to a logical error.
As a concrete example, consider the propagation of a two-qubit gate error in the $\ket{0}$-state preparation circuit of the first logical qubit when encoding into the Steane code:
\begin{align}
  \includegraphics[valign=c,scale=0.8]{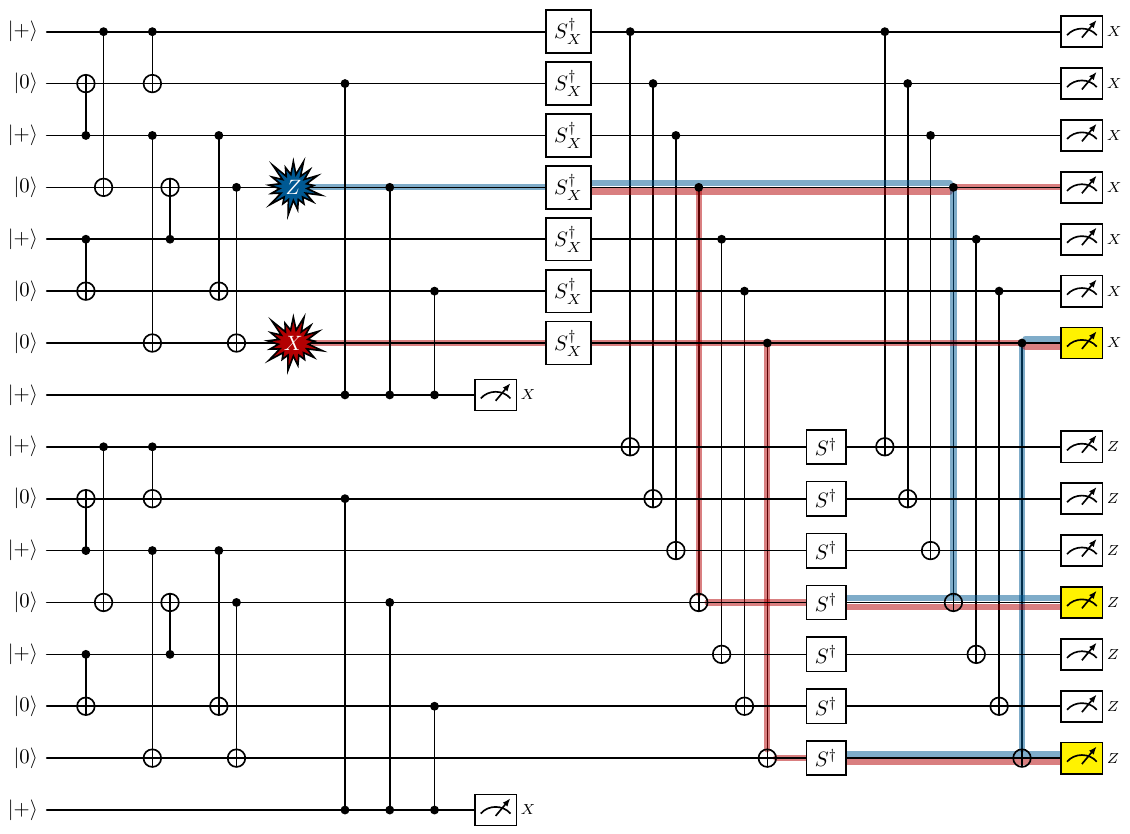}.
\end{align}
While the measurement error on the top logical qubit would be corrected by a naive decoder, the measurement errors on the bottom logical qubit would be decoded into a logical error, which occurs with (at least) the same probability as that of the propagated physical error.
This example illustrates that merely substituting physical gates by fault-tolerant gadgets is insufficient for a logical circuit to be fault tolerant, and that suppressing logical errors beyond $\mathcal{O}(p)$ therefore requires additional measures such as the insertion of QEC cycles to construct extended rectangles and the use of more sophisticated decoding strategies.

\section{Improved magic-state preparation}
\label{sec:prep_h}

\begin{figure}
  \includegraphics[scale=0.75]{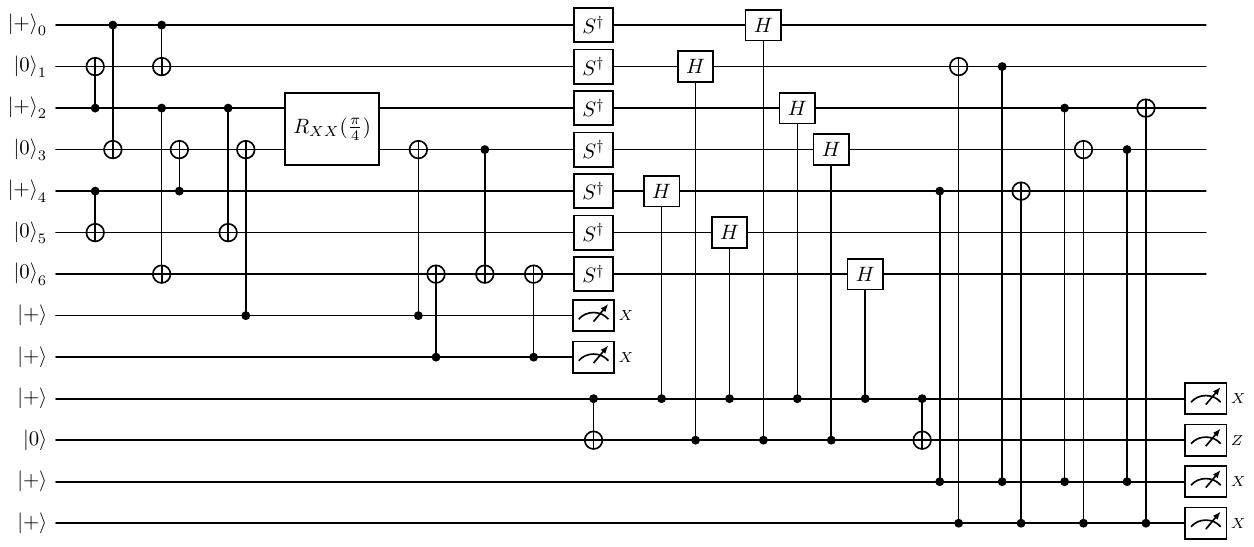}
  \caption{%
    A circuit for fault-tolerantly preparing logical $\ket{H}$ states with six ancilla qubits, 30 two-qubit gates, and a two-qubit gate depth of 9 (with as-soon-as-possible gate scheduling).
    This circuit was used for hardware benchmarks in \cref{sec:results}.
    An improved version of this circuit, discovered after our hardware experiments, is provided in \cref{fig:prep_h}\sref{fig:prep_h_new}, which has one fewer ancilla qubit and two fewer two-qubit gates.
  }
  \label{fig:prep_h_hw}
\end{figure}

Here we discuss the construction of our specialized circuit in \cref{fig:prep_h_hw} for fault-tolerantly preparing an $\ket{\overline{H}}$ state on hardware.
Additionally, we give a slightly improved variant of the circuit using one fewer ancilla and two fewer CNOT gates, which was identified after completing our experiments.

In Ref.~\cite{chamberland2019faulttolerant} the authors give a fault-tolerant preparation circuit for $\ket{\overline{H}}$ that is composed of three parts:
\begin{enumerate}
  \item\label{step:nonft_h}
    A non-fault-tolerant magic state preparation circuit.
  \item\label{step:meas_h}
    A flag fault-tolerant circuit for measuring the logical Hadamard operator.
  \item\label{step:qed}
    A full syndrome extraction cycle.
\end{enumerate}
We follow a similar overall approach but modify steps \ref{step:nonft_h} and \ref{step:qed} to minimize the number of two-qubit gates.
We reduce the number of two-qubit gates from 48 in the circuit given in Ref.~\cite{chamberland2019faulttolerant} to 30 for the circuit used in our experiments, and then to 28 in our improved version given here.

First we consider step \ref{step:nonft_h}, the non-fault-tolerant magic state preparation circuit.
As discussed in \cref{sec:steane_non_cliff} of the main text, our non-fault-tolerant magic state preparation circuit is constructed by back-propagating a representative of $\overline{X}$ until its support has weight 2.
This allows us to implement the logical rotation $R_{\overline{X}}(\theta) = \exp(-i\frac{\theta}{2}\overline{X})$ with a physical two-qubit gate, $R_{XX}(\theta) = \exp(-i\frac{\theta}{2} X\otimes X)$.
Next, following the approach of Ref.~\cite{forlivesi2025flag}, we flag this subcircuit to detect as many errors as possible at origin.
The circuit is not fully bipartite, so it cannot be fully flagged in this way, but many errors can be detected at this early stage.

Moving on to step \ref{step:qed}, we describe our final syndrome extraction gadgets.
We seek to find a minimal set of stabilizers such that if we measure them we will detect all errors arising from faults in the circuit.
To do this, we propagate all possible Pauli errors occurring after each gate in the circuit to the end of the circuit.
Since the circuit is not Clifford, this can produce linear combinations of Pauli strings at the end of the circuit.
For each propagated error, we exhaustively minimize its weight under the action of both the stabilizer group and the logical stabilizer $\overline{H}$.
If the propagated error has minimal weight less than 2, then it can be ignored.
In the case of linear combinations of Pauli strings, we apply weight reduction and filtering (after possible multiplication with the logical stabilizer $\overline{H}$) to each term of the superposition independently.
This procedure gives us a list of problematic Pauli strings that arise from some fault in the circuit and need to be detected.
We exhaustively compare this list to the full stabilizer group and find that there is no single element of the stabilizer group that anti-commutes with all problematic errors, but there are pairs of weight-4 stabilizers that anti-commute with all of them.
Out of these pairs, we choose one where the two stabilizers, one $X$ type and one $Z$ type, have the same support, as these can most efficiently be measured together.

Finally, we revisit step \ref{step:nonft_h} and attempt to remove any flag-at-origin gadgets that are no longer needed because they detect errors that are now also detected by the final syndrome measurements.
In this case, we are able to remove the majority of them.
The circuit in \cref{fig:prep_h_hw} retains two flags, whereas the circuit in \cref{fig:prep_h}\sref{fig:prep_h_new} retains only one.

\section{Flag-fault-tolerant QEC}
\label{sec:flag_qec}

\begin{figure}
  ~\hfill
  \subfloat[]
  {\label{fig:meas_xzz}\includegraphics[scale=0.75]{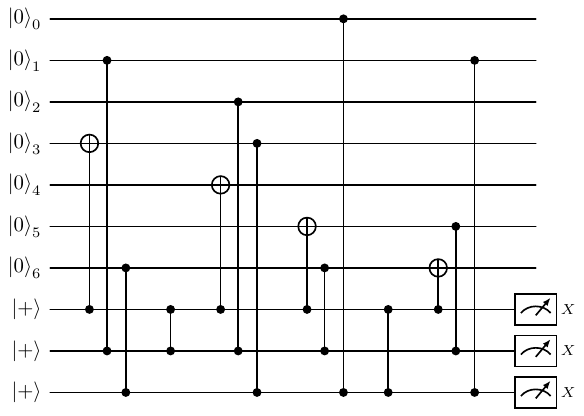}}
  \hfill
  \subfloat[]
  {\label{fig:meas_zxx}\includegraphics[scale=0.75]{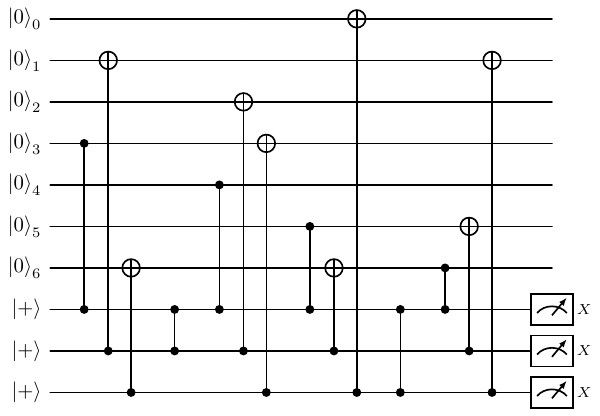}}
  \hfill~
  \caption{Subcircuits used for \sref{fig:meas_xzz} $XZZ$ and \sref{fig:meas_zxx} $ZXX$ stabilizer measurements in flagged QEC, adapted from Ref.~\cite{ryan-anderson2021realization}.}
  \label{fig:flag_qec}
\end{figure}

Here we describe, for completeness (that is, to provide explicit circuits and protocols for all gadgets used in this work), the flag-fault-tolerant QEC protocol in Ref.~\cite{ryan-anderson2021realization}.
This protocol is adapted slightly to be consistent with the choice of qubit ordering for the Steane code in this work.
The protocol is as follows.
\begin{enumerate}
  \item Perform a round of $XZZ$ stabilizer measurements using the flagged circuit in \cref{fig:flag_qec}\sref{fig:meas_xzz} and store the syndrome measurement outcomes, from top to bottom, to the bits $(x_1, z_2, z_3) \in \set{0, 1}^3$.
  \item If $(x_1, z_2, z_3) \ne (0, 0, 0)$, set $(z_1, x_2, x_3) = (0, 0, 0)$; otherwise, perform a round of $ZXX$ stabilizer measurements using the flagged circuit in \cref{fig:flag_qec}\sref{fig:meas_zxx} and store the syndrome measurement outcome to the bits $(z_1, x_2, x_3) \in \set{0, 1}^3$.
    If $x_1 = x_2 = x_3 = z_1 = z_2 = z_3 = 0$, stop the protocol without proceeding to the following steps.
  \item Perform a non-fault-tolerant bare-ancilla measurement of all $X$ stabilizer generators, corresponding to rows of the parity check matrix $h$ in \cref{eq:parity_check_matrix}, using a circuit analogous to that in \cref{fig:pauli_measurement_gadget}.
    Save the measurement outcome of the stabilizer with support given by the (first, second, third) rows of $h$ to the bits $(x_4, x_5, x_6)$.
  \item Similarly measure all $Z$ stabilizer generators and save the measurement outcomes to the bits $(z_4, z_5, z_6)$.
  \item Decode the syndromes $(x_4, x_5, x_6)$ and $(z_4, z_5, z_6)$, and perform active error correction as discussed in \cref{sec:steane_measurement}.
  \item Let
    \begin{align}
      E =
      \begin{pmatrix}
        1 & 0 & 0 & 0 & 1 & 1 \\
        1 & 0 & 0 & 0 & 0 & 1 \\
        0 & 1 & 1 & 0 & 0 & 1
      \end{pmatrix}.
    \end{align}
    If $(x_1, x_2, x_3, x_4, x_5, x_6) \in \rows(E)$, apply a logical $Z$ operator.
  \item If $(z_1, z_2, z_3, z_4, z_5, z_6) \in \rows(E)$, apply a logical $X$ operator.
\end{enumerate}
In words, the last two steps of this protocol correct hook errors that (i) may occur during flagged $XZZ$ and $ZXX$ syndrome readout, and (ii) propagate to logical errors.
Note that the protocol presented here is technically different from that in Ref.~\cite{ryan-anderson2021realization}, even after accounting for qubit ordering conventions, due to a different choice of stabilizer measurements for $(x_4, x_5, x_6)$ and $(z_4, z_5, z_6)$, which changes the conditions (that is, the matrix $E$) used to correct hook errors.
The protocol outlined here is used for \texttt{Helios} experiments, while the protocol in Ref.~\cite{ryan-anderson2021realization} (which is implemented in \texttt{PECOS}~\cite{ryan-anderson2018pecos}) was used for \texttt{H2-2} experiments.

\section{Details of the HHL implementation}
\label{sec:hhl_details}

Here we describe our compilation of the HHL algorithm to a QLSP that is suitable for near-term benchmarking.
For reference, the QLSP we consider is that of finding a least-squares solution to the Poisson equation
\begin{align}
  \nabla^2 \bm x = \bm b
\end{align}
for $\bm b = (1, -1, i, -i)$, which corresponds to the initial state $\ket{\bm b} = \ket{+Y}\otimes\ket{-X}$, where $\ket{+Y} = SH\ket{0}$ and $\ket{-X} = ZH\ket{0}$ are Pauli eigenstates.
The key components of the HHL circuit that we need to construct are a quantum phase estimation circuit for steps \ref{step:QPE} and \ref{step:QPE_inv} of the algorithm (as described in \cref{sec:hhl}), and an eigenvalue inversion circuit for step \ref{step:inversion}.
These components are combined to obtain the circuit in \cref{fig:circuit_hhl}\sref{fig:hhl_poisson}.
The final logical HHL circuit that we execute on \texttt{Helios} is provided in \cref{fig:hhl_compiled}.

\begin{figure}
  \includegraphics[width=\columnwidth]{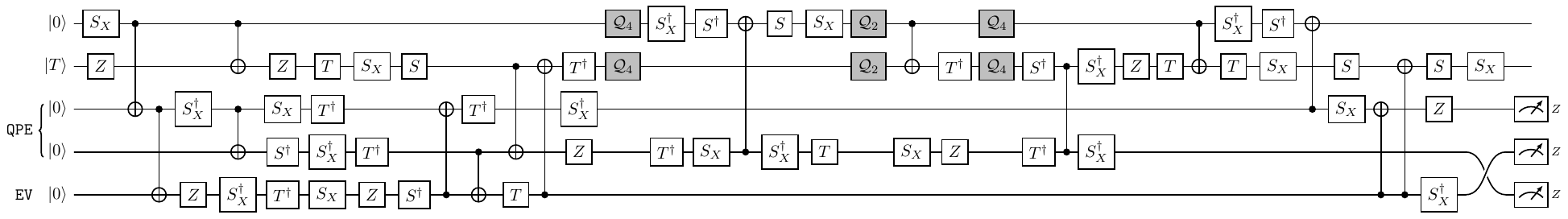}
  \caption{%
    The logical HHL circuit that we execute on \texttt{Helios}, obtained by compiling the circuit in \cref{fig:circuit_hhl}\sref{fig:hhl_poisson}.
    Grey boxes indicate locations at which QEC cycles are inserted in experiments with a total of $n_{\mathrm{QEC}}=2$ ($\mathcal{Q}_2$) or $n_{\mathrm{QEC}}=4$ ($\mathcal{Q}_4$) QEC cycles.
    An explicit SWAP on the last two qubits is added to track the identity of logical qubits, so that the inputs and outputs to the circuit are ordered identically to \cref{fig:circuit_hhl}\sref{fig:hhl_poisson}; in practice, this SWAP is performed virtually by relabeling qubits.
    In total, this circuit has one initial $\ket{T}$ state, 13 $T$ gates, and 15 two-qubit (CNOT or CZ) gates.
  }
  \label{fig:hhl_compiled}
\end{figure}

\subsection{Quantum phase estimation}

\begin{figure}
  \includegraphics[scale=0.75]{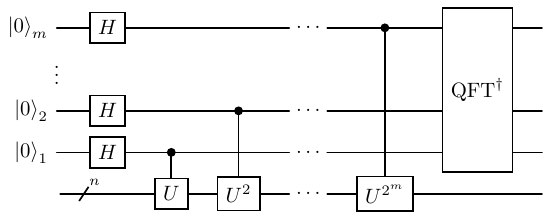}
  \caption{%
    The circuit for performing quantum phase estimation (QPE) to $m$ bits of precision, used in step \ref{step:QPE} of the HHL algorithm (\cref{sec:hhl}).
    Here $U = e^{i2\pi A}$.
  }
  \label{fig:qpe_full}
\end{figure}

Quantum phase estimation for step \ref{step:QPE} of the HHL algorithm is generally implemented to $m$ bits of precision by the circuit in \cref{fig:qpe_full}.
Judiciously normalizing $A = -\nabla^2 / 8$, we observe that the eigenvalues of $A$ are $\set{0,\frac12,\frac14}$, so two bits of precision suffice to compute the eigenvalues of $A$ exactly.
Letting $U = e^{i2\pi A}$, we observe that $A$ is diagonalized by the quantum Fourier transform (QFT), which means that
\begin{align}
  \includegraphics[valign=c, scale=0.75]{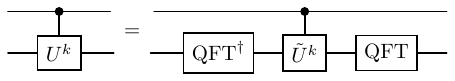},
\end{align}
where $\tilde U = e^{i2\pi\tilde A}$ is diagonal and $\tilde A = \mathrm{QFT}^\dag \cdot A \cdot \mathrm{QFT}$.
Specifically,
\begin{align}
  \tilde U
  = \diag(1, -i, -1, -i)
  = CX (S \otimes S) CX,
\end{align}
which implies that $\tilde U^2 = CX (Z\otimes Z) CX = I\otimes Z$, and in turn
\begin{align}
  \includegraphics[valign=c, scale=0.75]{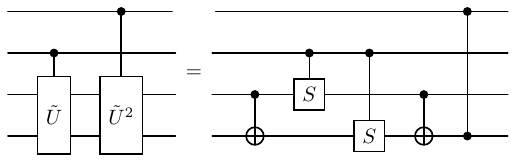}.
\end{align}
Here the controlled-$S$ gates can be implemented using 3 $T$ gates:
\begin{align}
  \includegraphics[valign=c, scale=0.75]{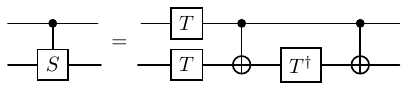}.
\end{align}

\subsection{Eigenvalue inversion}

Here we describe the eigenvalue inversion circuit that we use to implement the eigenvalue transformation in step \ref{step:inversion} of the HHL algorithm, namely
\begin{align}
  \ket{\lambda_j}
  \to \ket{\lambda_j} \otimes
  \sp{\tilde\lambda_j^{-1} \ket{0} + \sqrt{1 - \tilde\lambda_j^{-2}} \ket{1}}
  \label{eq:eigeninv}
\end{align}
where $j$ indexes an eigenvector (with eigenvalue $\lambda_j$) of the matrix $A$ that is to be inverted by HHL, and $\tilde\lambda^{-1}_j = C / \lambda$ for some constant $C$ that enforces $\abs{\tilde\lambda_j^{-1}}\le1$ (unless $\lambda_j=0$, for which we define $\tilde\lambda_j^{-1}=0$).

We consider the case in which eigenvalues $\lambda\in\set{0, \frac14, \frac12}$ are stored in a two-qubit register, $\ket{\lambda_j} = \ket{a_j}\otimes\ket{b_j}$, where we associate each value $\lambda_j \sim (a_j,b_j) \in \set{0,1}^2$ with a binary decimal encoding,
\begin{align}
  0 \sim (0, 0),
  &&
  \frac14 \sim (0, 1),
  &&
  \frac12 \sim (1, 0),
  &&
  \frac34 \sim (1, 1).
  \label{eq:eigenencoding}
\end{align}
Choosing $C = 1/4$, the desired transformations of a $\ket{0}$-state ancilla for eigenvalue inversion are
\begin{align}
  \renewcommand{\arraystretch}{1.5}
  \begin{array}{c|c|c|c}
    \lambda_j & (a_j, b_j) & \tilde\lambda_j^{-1} & \mathrm{final~ancilla~state}
    \\ \hline \hline
    0 & (0, 0) & 0 & \ket{1}
    \\
    \frac14 & (0, 1) & 1 & \ket{0}
    \\
    \frac12 & (1, 0) & 1/2 & \frac{1}{2}\ket{0} + \frac{\sqrt{3}}{{2}}\ket{1}
  \end{array}
\end{align}
where the target ancilla state for $\lambda_j=\frac12$ is $R_Y(\frac{2\pi}{3}) \ket{0}$.
These transformations are implemented coherently by the circuit
\begin{align}
  \includegraphics[valign=c, scale=0.75]{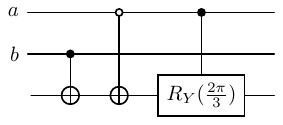}.
\end{align}
Using the Solovay-Kitaev algorithm~\cite{kitaev1997quantum, dawson2005solovaykitaev}, we find that controlled-$R_Y(\frac{2\pi}{3})$ gate can be approximated to a channel fidelity of $\sim0.966$ by the Clifford+$T$ decomposition
\begin{align}
  \includegraphics[valign=c, scale=0.75]{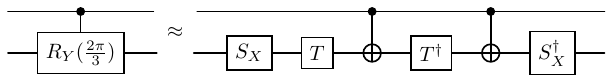},
\end{align}
where $S_X = R_X(\frac{\pi}{2}) = HSH$.

\section{The Ramsey protocol with a dephasing noise model}
\label{sec:ramsey}

Here we describe a single-parameter dephasing model, which we call \texttt{DepModel} in the main text, of a Ramsey protocol in which repeated $T$ gates are applied to an initial $\ket{+}$ state.
For a fixed error parameter $p$, this model predicts that the Ramsey fidelity $\mathcal{F}_{T,p}(n)$ (that is, the fidelity of the noisily prepared state with respect to the noiseless target state) decays exponentially with respect to the number of applied $T$ gates, $n$.

After Pauli twirling, a noisy $T$ gate can be described by the single-qubit channel~\cite{piveteau2021error}
\begin{align}
  \mathcal{T}_p = (1 - p) \mathcal{T} + p \mathcal{D}_\mathrm{z}
  \label{eq:noisy_T}
\end{align}
where $\mathcal{T}$ is the channel for a noiseless $T$ gate and $\mathcal{D}_\mathrm{z}$ is a dephasing channel, which act on a single-qubit density matrix $\sigma$ as
\begin{align}
  \mathcal{T}(\sigma) = T\sigma T^\dag,
  &&
  \mathcal{D}_\mathrm{z}(\sigma)
  = \frac12(\sigma + Z\sigma Z)
  = \sigma_{00}\op{0} + \sigma_{11} \op{1}.
\end{align}
Here $\sigma_{ab} = \braket{a|\sigma|b}$, and the sum of channels $A$ and $B$ acts on a density matrix as $(A+B)(\rho) = A(\rho) + B(\rho)$.
The parameter $p\in[0,1]$ can be interpreted as the probability of a dephasing event in which all relative phase information is lost---in which case $\sigma$ is replaced by $\mathcal{D}_\mathrm{z}(\sigma)$---when applying a noisy $T$ gate.
Note that the definition of the noisy $T$ gate in \cref{eq:noisy_T} is slightly different from that in Eq.~(5) of Ref.~\cite{mayer2024benchmarking}; the definition used here is chosen in order to leverage properties of the dephasing channel $\mathcal{D}_\mathrm{z}$ below.

The primary claim of this section is the following:
\begin{lemma}
  If $\sigma$ is the density matrix of a pure state that is an equal superposition of $\ket{0}$ and $\ket{1}$, for which $\sigma_{00} = \sigma_{11} = 1/2$, then the fidelity of the noisy state $\mathcal{T}_p^n(\sigma)$ with respect to the noiseless target state $\mathcal{T}^n(\sigma) = T^n \sigma {T^\dag}^n$ is
  \begin{align}
    \mathcal{F}_{T,p}(n) = \frac12 + \frac12 (1 - p)^n.
    \label{eq:fidelity_T_dep}
  \end{align}
  \label{lemma:dephasing}
\end{lemma}
\begin{proof}
  Observe that:
  \begin{enumerate}[label=(\alph*)]
    \item The dephasing channel deletes phase information, meaning that $\mathcal{T} \circ \mathcal{D}_\mathrm{z} = \mathcal{D}_\mathrm{z} \circ \mathcal{T} = \mathcal{D}_\mathrm{z}$.
    \item The dephasing channel is a projection, meaning that $\mathcal{D}_\mathrm{z}^2 = \mathcal{D}_\mathrm{z}$.
  \end{enumerate}
  Together, these facts allow us to expand an $n$-fold application of the noisy $T$ gate as
  \begin{align}
    \mathcal{T}_p^n
    = (1-p)^n \mathcal{T}^n
    + [1-(1-p)^n] \mathcal{D}_\mathrm{z}.
    \label{eq:T_p_n}
  \end{align}
  If $\sigma$ is a pure state, then $\mathcal{T}^n(\sigma)$ is a pure state, whose fidelity $\mathcal{F}_{T,p}(n)$ with the noisy state $\mathcal{T}_p^n(\sigma)$ is
  \begin{align}
    \mathcal{F}_{T,p}(n)
    &= \mathrm{Tr}[\mathcal{T}_p^n(\sigma) \mathcal{T}^n(\sigma)] \\
    &= (1 - p)^n \mathrm{Tr}[\mathcal{T}^n(\sigma) \mathcal{T}^n(\sigma)]
    + [1-(1-p)^n] \mathrm{Tr}[\mathcal{D}_\mathrm{z}(\sigma) \mathcal{T}^n(\sigma)] \\
    &= (1 - p)^n
    + [1-(1-p)^n] (\sigma_{00}^2 + \sigma_{11}^2).
  \end{align}
  Substituting $\sigma_{00} = \sigma_{11} = 1/2$ yields $\mathcal{F}_{T,p}(n) = \frac12 + \frac12 (1 - p)^n$.
\end{proof}
Note that the infidelity of a single noisy $T$ gate acting on the initial state $\op{+}$ is $\epsilon = 1 - \mathcal{F}_{T,p}(1) = p/2$.
Substituting $p=2\epsilon$ into Eq.~\eqref{eq:fidelity_T_dep} recovers the expression in Eq.~(6) of Ref.~\cite{mayer2024benchmarking}.

\section{Non-exponential decay of the Ramsey fidelity}
\label{sec:non_exponential}

The dephasing model in \cref{sec:ramsey} assumes that logical $T$ gates cause independent and identically distributed (IID) logical errors, and consequently predicts an exponential decay of the Ramsey fidelity.
This model agrees with simulations of Ramsey sequences in which every $T$ gate is followed by a QEC cycle, but fails to capture the non-exponential decay of Ramsey fidelity for sequences of $T$ gates without QEC cycles (\cref{fig:t_gate_sim}).
Fundamentally, this non-exponential decay is due the fact that the logical $T$ gates considered in this work are \emph{transparent} to $Z$ errors.
That is, a physical $Z$ error on a data qubit at the input to the logical $T$ gate gadget in \cref{fig:t_gate_and_qec}\sref{fig:qec_steane_swap} propagates to a physical $Z$ error at the output of this gadget.
A logical $Z$ error can therefore be caused by pairs of physical $Z$ errors that occur at distant $T$ gates.
Said otherwise, the entire Ramsey sequence can tolerate only one data-qubit $Z$ error.

Consider a Ramsey sequence of $n$ logical $T$ gates, in which every logical $T$ dephases each data qubit with IID probability $g$.
By a derivation analogous to that of \cref{eq:T_p_n} in \cref{sec:ramsey}, the probability that a given data qubit has been dephased at the end of the sequence is
\begin{align}
  r = 1 - (1 - g)^n \approx gn.
  \label{eq:prob_r}
\end{align}
Recalling that the dephasing channel is $D_{\mathrm{z}}(\sigma) = \frac12 (\sigma + Z\sigma Z)$, a probability $r$ of a dephasing event corresponds to a probability $\gamma = r/2$ of a $Z$ error ($\sigma \to Z \sigma Z$).
The probability $\Gamma$ of an uncorrectable $Z$ error at the end of a Ramsey sequence can be found by enumerating all configurations of uncorrectable multi-qubit phase errors in the Steane code:
\begin{align}
  \Gamma = {7 \choose 2} \gamma^2 (1 - \gamma)^5 + 7 \gamma^3 (1 - \gamma)^4
  + (7\times 4) \gamma^4 (1 - \gamma)^3 + 7 \gamma^6 (1-\gamma) + \gamma^7
  \approx 21 \gamma^2 \approx \frac{21}{4} g^2 n^2.
\end{align}
In turn, the probability of a logical dephasing event at the end of the Ramsey sequence is $P = 2\Gamma$.
Additionally accounting for an IID probability $p$ of a logical dephasing event at each logical $T$ gate, the overall Ramsey fidelity becomes
\begin{align}
  \mathcal{F}_{\text{Ramsey}}(p,g,n)
  = \frac12 + \frac12 (1 - P) (1 - p)^n
  \approx \frac12 + \frac12 \exp\left(-p n - \frac{21}{2} g^2 n^2\right).
  \label{eq:ramsey_gaussian}
\end{align}
Figure \cref{fig:t_gate_sim_appendix} shows the fidelity of Ramsey sequences without QEC cycles (same as the ``without QEC'' data in \cref{fig:t_gate_sim} of the main text), together with fits to both the full expression (in terms of $P=2\Gamma$) and the Gaussian approximation in \cref{eq:ramsey_gaussian}.
Notably, the full expression in \cref{eq:ramsey_gaussian}, written as a function of $(g,p,n)$ without taking any of the approximations in \crefrange{eq:prob_r}{eq:ramsey_gaussian}, provides a markedly better fit to early-time behavior than the Gaussian approximation in \cref{eq:ramsey_gaussian}.
We use the full expression in \cref{eq:ramsey_gaussian} for the fit in \cref{fig:t_gate_sim} of the main text.

\begin{figure}
  \includegraphics{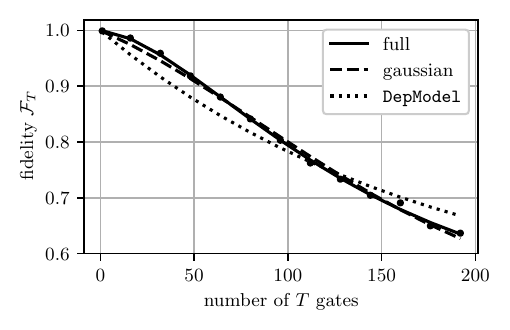}
  \caption{%
    Fidelity of Ramsey sequences without QEC cycles, and fits to the single-parameter \texttt{DepModel} in \cref{eq:fidelity_T_dep} (dotted line), the two-parameter Gaussian approximation in \cref{eq:ramsey_gaussian} (dashed line), and the full two-parameter model in \cref{eq:ramsey_gaussian} (solid line).
    The simulation data here is the same as the ``without QEC'' data in \cref{fig:t_gate_sim}.
  }
  \label{fig:t_gate_sim_appendix}
\end{figure}

\section{Entanglement fidelity of a single-qubit channel}
\label{sec:channel_fidelity}

The entanglement fidelity of a single-qubit channel $\mathcal{C}$ is~\cite{nielsen1996entanglement}
\begin{align}
  \mathcal{F}(\mathcal{C})
  = \braket{\phi|(\mathcal{I}\otimes\mathcal{C})(\phi)|\phi}
  = \Tr\sp{(\mathcal{I}\otimes\mathcal{C})(\phi)\phi}
\end{align}
where $\ket\phi = \frac1{\sqrt2} \sum_{a\in\set{0,1}} \ket{aa}$ is a maximally entangled state of two qubits and $\phi = \op{\phi}$.
Unlike the commonly used gate fidelity, the entanglement fidelity penalizes the deleterious effect of gate errors on the entanglement between the target qubits of a gate and other qubits in a quantum computation~\cite{nielsen1996entanglement}.
For reference, the entanglement fidelity $\mathcal{F}_{\mathcal{G}}$ of a noisy channel $\mathcal{G}$ acting on a $d$-dimensional system ($d=2$ for a qubit) is related to the gate fidelity $\mathcal{F}_{\mathcal{G}}^{\mathrm{gate}}$ by $\mathcal{F}_{\mathcal{G}}^{\mathrm{gate}} = (d\mathcal{F}_{\mathcal{G}} + 1) / (d+1)$~\cite{horodecki1999general, nielsen2002simple}.

\begin{lemma}
  The entanglement fidelity of the single-qubit channel $\mathcal{C}$ can be written as
  \begin{align}
    \mathcal{F}(\mathcal{C})
    = \frac14 \sum_{\substack{P\in\set{X,Y,Z}\\b\in\set{0,1}}}
    \Tr\sp{\mathcal{C}(\Pi_{P,b}) \Pi_{P,b}} - \frac12,
  \end{align}
  where $\Pi_{P,b} = (I + (-1)^b P)/2$ is the projector into the eigenstate of the Pauli operator $P$ with eigenvalue $(-1)^b$, for which $P = \Pi_{P,0} - \Pi_{P,1}$.
  \label{lemma:entanglement}
\end{lemma}
\begin{proof}
  By direct substitution,
  \begin{align}
    \mathcal{F}(\mathcal{C})
    &= \frac14 \sum_{a,b,c,d\in\set{0,1}}
    \Tr\sp{(\op{a}{b}\otimes\mathcal{C}(\op{a}{b})) (\op{c}{d}\otimes\op{c}{d})} \\
    &= \frac14 \sum_{a,b,c,d\in\set{0,1}}
    \Tr\sp{\op{a}{b}\op{c}{d}} \times \Tr\sp{\mathcal{C}(\op{a}{b})\op{c}{d}},
  \end{align}
  where we used the fact that $\Tr(A\otimes B) = \Tr(A) \times \Tr(B)$.
  Here $\Tr\sp{\op{a}{b}\op{c}{d}} = 1$ if $(c,d) = (b,a)$ and $0$ otherwise, so
  \begin{align}
    \mathcal{F}(\mathcal{C})
    = \frac14 \sum_{a,b\in\set{0,1}} \Tr\sp{\mathcal{C}(\op{a}{b})\op{b}{a}}.
  \end{align}
  The sum over $\op{a}{b}$ can be replaced by a sum over any orthonormal basis $\mathcal{B}_2$ for the space of linear operators on a single-qubit Hilbert space, which is to say that
  \begin{align}
    \mathcal{F}(\mathcal{C})
    = \frac14 \sum_{M\in\mathcal{B}_2} \Tr\sp{\mathcal{C}(M) M^\dag}
  \end{align}
  where $\mathcal{B}_2$ is any basis of $2\times2$ matrices for which $A,B\in\mathcal{B}_2$ implies $\Tr(A B^\dag) = 1$ iff $A=B$ and $0$ otherwise.
  Considering in particular the Pauli basis $\set{I,X,Y,Z}/\sqrt{2}$, we use the facts that $\mathcal{C}$ is a trace-preserving map, meaning $\Tr\sp{\mathcal{C}(I)} = 2$, and expand $P = \Pi_{P,0} - \Pi_{P,1} = \sum_{b\in\set{0,1}} (-1)^b \Pi_{P,b}$ for all $P\in\set{X,Y,Z}$ to write
  \begin{align}
    \mathcal{F}(\mathcal{C})
    &= \frac18 \Tr\sp{\mathcal{C}(I)} + \frac18 \sum_{P\in\set{X,Y,Z}} \Tr\sp{\mathcal{C}(P) P} \\
    &= \frac14 + \frac18 \sum_{\substack{P\in\set{X,Y,Z}\\b\in\set{0,1}}}
    (-1)^b \Tr\sp{\mathcal{C}(\Pi_{P,b}) P}.
  \end{align}
  Substituting $(-1)^b P = 2 \Pi_{P,b} - I$, we then find that
  \begin{align}
    \mathcal{F}(\mathcal{C})
    = \frac14 + \frac14 \sum_{\substack{P\in\set{X,Y,Z}\\b\in\set{0,1}}}
    \Tr\sp{\mathcal{C}(\Pi_{P,b}) \Pi_{P,b}}
    - \frac18 \sum_{\substack{P\in\set{X,Y,Z}\\b\in\set{0,1}}} \Tr\sp{\mathcal{C}(\Pi_{P,b})},
  \end{align}
  where $\Tr\sp{\mathcal{C}(\Pi_{P,b})} = \Tr(\Pi_{P,b}) = 1$ because $\mathcal{C}$ is trace-preserving, so the final sum $\sum_{P,b}\Tr\sp{\mathcal{C}(\Pi_{P,b})}=6$, and altogether
  \begin{align}
    \mathcal{F}(\mathcal{C})
    = \frac14 \sum_{\substack{P\in\set{X,Y,Z}\\b\in\set{0,1}}}
    \Tr\sp{\mathcal{C}(\Pi_{P,b}) \Pi_{P,b}} - \frac12.
    \label{eq:entanglement_final}
  \end{align}
\end{proof}
\begin{lemma}
  If $\mathcal{Q}_p = (1-p)\mathcal{I} + p \mathcal{D}$, where $p\in[0,1]$ is a scalar and $\mathcal{D}$ is a single-qubit channel that maps any input states to the maximally mixed state $I/2$, then
  \begin{align}
    \mathcal{F}(\mathcal{Q}_p^n)
    = \frac14 + \frac34 (1-p)^n.
  \end{align}
\end{lemma}
\begin{proof}
  Let $\mathcal{Q} = (1-p)\mathcal{I} + p\mathcal{D}$.
  Similarly to the proof of \cref{lemma:dephasing}, observe that $\mathcal{I}\circ\mathcal{D} = \mathcal{D}\circ\mathcal{I} = \mathcal{D}^2 = \mathcal{D}$, so
  \begin{align}
    \mathcal{Q}_p^n = (1-p)^n \mathcal{I} + \sp{1 - (1-p)^n} \mathcal{D},
  \end{align}
  which implies that
  \begin{align}
    \Tr\sp{\mathcal{Q}_p^n(\Pi_{P,b}) \Pi_{P,b}}
    = (1-p)^n \times \Tr\sp{\Pi_{P,b}^2} + \sp{1 - (1-p)^n} \times \frac12 \Tr\sp{\Pi_{P,b}}
    = \frac12 + \frac12 (1-p)^n.
  \end{align}
  By direct substitution into \cref{eq:entanglement_final} (\cref{lemma:entanglement}),
  \begin{align}
    \mathcal{F}(\mathcal{Q}_p^n) = \frac14 + \frac34 (1-p)^n.
  \end{align}
\end{proof}

\section{Additional benchmarking data}
\label{sec:data_tables}

Tables \ref{tab:labs_5}--\ref{tab:hhl_data} provide detailed emulator and hardware data collected for the application benchmarks presented in \cref{sec:results} of the main text.
Data tables indicate the number of logical qubits in an algorithm ($N$), the number of $T$ gates in the executed circuit after decomposition to a Clifford+$T$ gate set ($n_T$), and, if applicable, the QAOA depth $p$.
For reference, data tables also provide observables of interest for noise-free (NF) circuits both before (original) and after (decomposed) decomposition to a Clifford+$T$ gate set.
Results for unencoded (physical, decomposed) QAOA circuits have ``N/A'' for their RUS limit and number of QEC cycles, while all other circuits are encoded into the Steane code.
QEC cycles are performed using the \texttt{Steane-swap} gadget in \cref{fig:t_gate_and_qec}\sref{fig:qec_steane_swap}, for which logical $\ket{0}$ and $\ket{+} = H\ket{0}$ states are prepared using the circuit in \cref{fig:prep_z}.
Logical $T$ gates are implemented using the \texttt{T-swap} gadget in \cref{fig:t_gate_and_qec}\sref{fig:t_gate_swap}, for which logical $\ket{T} = HS\ket{H}$ states are prepared using the circuit in \cref{fig:prep_h_hw}.
``Total shots'' refers to the total number of shots executed on hardware (or an emulator), while ``kept shots'' refers to the number of shots in which all state preparation subroutines succeeded (indicated by zero measurement outcomes on the appropriate flag/ancilla qubits).

The QAOA data tables also report the LABS \emph{merit factor} of a bitstring $\bm s\in\set{+1,-1}^N$, which is defined as
\begin{align}
  F_{\mathrm{LABS}}(\bm s) = \frac{N^2}{2 E_{\mathrm{LABS}}(\bm s)},
\end{align}
and is conjectured to approach $F_{\mathrm{LABS}}^{\mathrm{opt}} \approx 12.32$ as $N\to\infty$~\cite{golay1982merit}.
The merit factor of a QAOA state $\ket{\bm\beta,\bm\gamma}$ is
\begin{align}
  F_{\mathrm{LABS}}(\bm\beta,\bm\gamma)
  = \frac{N^2}{2\braket{\bm\beta,\bm\gamma|H_{\mathrm{LABS}}|\bm\beta,\bm\gamma}}.
  \label{eq:LABS_merit}
\end{align}

\begin{sidewaystable}[h]
  \centering
  \caption{%
    Results from running QAOA for $N=5$ LABS on \texttt{H2-2} and \texttt{Helios} computers, as well as the \texttt{H2-2E} emulator.
    Circuits with QEC cycles apply one QEC cycle to each logical qubit in the middle of the logical circuit, evenly partitioning the logical two-qubit gates.
  }
  \vspace{.5em}
  \setlength{\tabcolsep}{2pt}
  \renewcommand{\arraystretch}{1.2}
  \begin{tabular}{|l||c||c||*{2}{c|}|*{4}{c|}|*{3}{c|}}
    \hline
    {\textbf{logical qubits}}
    & \multicolumn{11}{c|}{$N = 5$}
    \\ \hline
    {\textbf{circuit parameters}}
    & \multicolumn{2}{c||}{$(p, n_T) = (1, 9)$}
    & \multicolumn{2}{c||}{$(p, n_T) = (2, 5)$}
    & \multicolumn{7}{c|}{$(p, n_T) = (2, 9)$}
    \\ \hline \hline
    \textbf{success probability (original, NF)}
    & \multicolumn{2}{c||}{0.782}
    & \multicolumn{2}{c||}{0.935}
    & \multicolumn{7}{c|}{0.935}
    \\ \hline
    \textbf{success probability (decomposed, NF)}
    & \multicolumn{2}{c||}{0.694}
    & \multicolumn{2}{c||}{0.625}
    & \multicolumn{7}{c|}{0.854}
    \\ \hline
    \textbf{merit factor (original, NF)}
    & \multicolumn{2}{c||}{5.174}
    & \multicolumn{2}{c||}{5.931}
    & \multicolumn{7}{c|}{5.931}
    \\ \hline
    \textbf{merit factor (decomposed, NF)}
    & \multicolumn{2}{c||}{4.641}
    & \multicolumn{2}{c||}{4.375}
    & \multicolumn{7}{c|}{5.396}
    \\ \hline \hline
    \textbf{device}
    & \multicolumn{1}{c||}{\texttt{H2-2E}}
    & \multicolumn{1}{c||}{\texttt{H2-2}}
    & \multicolumn{2}{c||}{\texttt{H2-2}}
    & \multicolumn{4}{c||}{\texttt{H2-2}}
    & \multicolumn{3}{c|}{\texttt{Helios}}
    \\ \hline
    \textbf{RUS limit}
    & N/A & 1
    & N/A & 1
    & N/A & 1 & 1 & 2
    & 1 & 1 & 2
    \\ \hline
    \textbf{QEC cycles}
    & N/A & 0
    & N/A & 0
    & N/A & 0 & 5 & 0
    & 0 & 5 & 0
    \\ \hline
    \textbf{physical qubits}
    & 5 & 48
    & 5 & 48
    & 5 & \multicolumn{3}{c||}{48}
    & \multicolumn{3}{c|}{48}
    \\ \hline
    \textbf{physical 2-qubit gates}
    & 11 & 474
    & 22 & 399
    & 22 & 551 & 731 & $\le 885$
    & 551 & 731 & $\le 885$
    \\ \hline
    \textbf{estimated NF runtime (seconds)}
    & -- & --
    & -- & --
    & -- & -- & -- & --
    & 1.37 & 1.92 & --
    \\ \hline \hline
    \textbf{success probability}
    & 0.522(22) & 0.466(20)
    & 0.484(22) & 0.453(29)
    & 0.592(22) & 0.587(19) & 0.491(31) & 0.500(23)
    & 0.545(19) & 0.567(19) & 0.544(22)
    \\ \hline
    \textbf{merit factor}
    & 4.65(11) & 4.25(10)
    & 4.43(11) & 4.15(15)
    & 5.34(09) & 4.97(09) & 4.39(16) & 4.53(11)
    & 4.83(10) & 4.83(10) & 4.74(11)
    \\ \hline
    \textbf{kept/total shots}
    & 500/500 & 625/1000
    & 500/500 & 285/400
    & 500/500 & 664/1000 & 267/500 & 488/500
    & 662/1000 & 652/1000 & 493/500
    \\ \hline
    \textbf{post-selection rate}
    & 1 & 0.625
    & 1 & 0.7125
    & 1 & 0.664 & 0.534 & 0.976
    & 0.662 & 0.652 & 0.986
    \\ \hline
  \end{tabular}
  \label{tab:labs_5}
\end{sidewaystable}

\begin{sidewaystable}[h]
  \centering
  \caption{%
    Results from running QAOA for $N=6$ LABS on \texttt{H2-2} and \texttt{Helios} computers, as well as the \texttt{H2-2E} emulator.
    Circuits with QEC cycles apply one QEC cycle to each logical qubit in the middle of the logical circuit, evenly partitioning the logical two-qubit gates.
  }
  \vspace{.5em}
  \setlength{\tabcolsep}{2pt}
  \renewcommand{\arraystretch}{1.2}
  \begin{tabular}{|l||c|c||c|c||c||c||c||c|c|c|}
    \hline
    {\textbf{logical qubits}}
    & \multicolumn{10}{c|}{$N = 6$}
    \\ \hline
    {\textbf{circuit parameters}}
    & \multicolumn{4}{c||}{$(p, n_T) = (1, 13)$}
    & \multicolumn{2}{c||}{$(p, n_T) = (1, 6)$}
    & \multicolumn{4}{c|}{$(p, n_T) = (2, 0)$}
    \\ \hline \hline
    \textbf{success probability (original, NF)}
    & \multicolumn{4}{c||}{0.915}
    & \multicolumn{2}{c||}{0.915}
    & \multicolumn{4}{c|}{1}
    \\ \hline
    \textbf{success probability (decomposed, NF)}
    & \multicolumn{4}{c||}{0.979}
    & \multicolumn{2}{c||}{0.672}
    & \multicolumn{4}{c|}{1}
    \\ \hline
    \textbf{merit factor (original, NF)}
    & \multicolumn{4}{c||}{2.438}
    & \multicolumn{2}{c||}{2.438}
    & \multicolumn{4}{c|}{2.571}
    \\ \hline
    \textbf{merit factor (decomposed, NF)}
    & \multicolumn{4}{c||}{2.523}
    & \multicolumn{2}{c||}{1.998}
    & \multicolumn{4}{c|}{2.571}
    \\ \hline \hline
    \textbf{device}
    & \multicolumn{2}{c||}{\texttt{H2-2}}
    & \multicolumn{2}{c||}{\texttt{Helios}}
    & \multicolumn{1}{c||}{\texttt{H2-2E}}
    & \multicolumn{1}{c||}{\texttt{H2-2}}
    & \multicolumn{1}{c||}{\texttt{H2-2}}
    & \multicolumn{3}{c|}{\texttt{Helios}}
    \\ \hline
    \textbf{RUS limit}
    & N/A & 1
    & 1 & 2
    & N/A
    & 1
    & 1
    & N/A & 1 & 1
    \\ \hline
    \textbf{QEC cycles}
    & N/A & 0
    & 0 & 0
    & N/A
    & 0
    & 0
    & N/A & 0 & 6
    \\ \hline
    \textbf{physical qubits}
    & 6 & 55
    & \multicolumn{2}{c||}{55}
    & 6 & 55
    & 55
    & 6 & \multicolumn{2}{c|}{55}
    \\ \hline
    \textbf{physical 2-qubit gates}
    & 18 & 686
    & 686 & $\le 1155$
    & 18
    & 420
    & 318
    & 36 & 318 & 618
    \\ \hline
    \textbf{estimated NF runtime (seconds)}
    & -- & --
    & 1.78 & --
    & --
    & --
    & --
    & 0.19 & 0.64 & 1.33
    \\ \hline \hline
    \textbf{success probability}
    & 0.957(12) & 0.784(28)
    & 0.562(33) & 0.628(22)
    & 0.668(21)
    & 0.582(27)
    & 0.893(15)
    & 0.982(06) & 0.935(11) & 0.860(17)
    \\ \hline
    \textbf{merit factor}
    & 2.484(24) & 2.252(42)
    & 1.927(50) & 2.017(34)
    & 1.985(38)
    & 1.874(47)
    & 2.417(22)
    & 2.546(08) & 2.481(16) & 2.378(23)
    \\ \hline
    \textbf{kept/total shots}
    & 300/300 & 218/500
    & 226/500 & 479/500
    & 500/500
    & 323/500
    & 422/500
    & 500/500 & 474/500 & 435/500
    \\ \hline
    \textbf{post-selection rate}
    & 1 & 0.436
    & 0.452 & 0.958
    & 1
    & 0.646
    & 0.844
    & 1 & 0.948 & 0.870
    \\ \hline
  \end{tabular}
  \label{tab:labs_6}
\end{sidewaystable}

\begin{table}[h]
  \centering
  \caption{%
    Results from running QAOA for $N=7,8,9$ LABS on the \texttt{Helios} computer.
  }
  \vspace{.5em}
  \setlength{\tabcolsep}{2pt}
  \renewcommand{\arraystretch}{1.2}
  \begin{tabular}{|l||c|c||c|c||c|c|}
    \hline
    {\textbf{logical qubits}}
    & \multicolumn{2}{c||}{$N = 7$}
    & \multicolumn{2}{c||}{$N = 8$}
    & \multicolumn{2}{c|}{$N = 9$}
    \\ \hline
    {\textbf{circuit parameters}}
    & \multicolumn{2}{c||}{$(p, n_T) = (1, 0)$}
    & \multicolumn{2}{c||}{$(p, n_T) = (1, 8)$}
    & \multicolumn{2}{c|}{$(p, n_T) = (1, 25)$}
    \\ \hline \hline
    \textbf{success probability (original, NF)}
    & \multicolumn{2}{c||}{0.118}
    & \multicolumn{2}{c||}{0.170}
    & \multicolumn{2}{c|}{0.122}
    \\ \hline
    \textbf{success probability (decomposed, NF)}
    & \multicolumn{2}{c||}{0.125}
    & \multicolumn{2}{c||}{0.125}
    & \multicolumn{2}{c|}{0.108}
    \\ \hline
    \textbf{merit factor (original, NF)}
    & \multicolumn{2}{c||}{2.792}
    & \multicolumn{2}{c||}{1.806}
    & \multicolumn{2}{c|}{1.862}
    \\ \hline
    \textbf{merit factor (decomposed, NF)}
    & \multicolumn{2}{c||}{2.880}
    & \multicolumn{2}{c||}{1.728}
    & \multicolumn{2}{c|}{1.764}
    \\ \hline \hline
    \textbf{device}
    & \multicolumn{6}{c|}{\texttt{Helios}}
    \\ \hline
    \textbf{RUS limit}
    & N/A & 1
    & N/A & 1
    & N/A & 1
    \\ \hline
    \textbf{QEC cycles}
    & N/A & 0
    & N/A & 0
    & N/A & 0
    \\ \hline
    \textbf{physical qubits}
    & 7 & 62
    & 8 & 69
    & 9 & 76
    \\ \hline
    \textbf{physical 2-qubit gates}
    & 37 & 336
    & 49 & 735
    & 78 & 1595
    \\ \hline
    \textbf{estimated NF runtime (seconds)}
    & 0.17 & 0.69
    & 0.21 & 1.69
    & 0.29 & 3.78
    \\ \hline \hline
    \textbf{success probability}
    & 0.136(15) & 0.109(14)
    & 0.102(14) & 0.108(17)
    & 0.130(15) & 0.061(20)
    \\ \hline
    \textbf{merit factor}
    & 2.998(101) & 2.770(97)
    & 1.710(47) & 1.724(58)
    & 1.842(40) & 1.682(20)
    \\ \hline
    \textbf{kept/total shots}
    & 500/500 & 469/500
    & 500/500 & 332/500
    & 500/500 & 148/500
    \\ \hline
    \textbf{post-selection rate}
    & 1 & 0.938
    & 1 & 0.664
    & 1 & 0.296
    \\ \hline
  \end{tabular}
  \label{tab:labs_hero_run}
\end{table}

\begin{table}[h]
  \centering
  \caption{%
    Results from running QAOA on a selected $N=12$ portfolio optimization instance on the \texttt{Helios} computer.
  }
  \vspace{.5em}
  \setlength{\tabcolsep}{2pt}
  \renewcommand{\arraystretch}{1.2}
  \begin{tabular}{|l||c|c||c|c|}
    \hline
    {\textbf{logical qubits}}
    & \multicolumn{4}{c|}{$N = 12$}
    \\ \hline
    {\textbf{circuit parameters}}
    & \multicolumn{2}{c||}{$(p, n_T) = (1, 2)$}
    & \multicolumn{2}{c|}{$(p, n_T) = (2, 4)$}
    \\ \hline \hline
    \textbf{success probability (original, NF)}
    & \multicolumn{2}{c||}{0.2672}
    & \multicolumn{2}{c|}{0.7077}
    \\ \hline
    \textbf{success probability (decomposed, NF)}
    & \multicolumn{2}{c||}{0.1821}
    & \multicolumn{2}{c|}{0.2124}
    \\ \hline
    \textbf{success probability (random guess)}
    & \multicolumn{4}{c|}{$1/2^{12}\approx0.00025$}
    \\ \hline
    \textbf{approximation ratio (original, NF)}
    & \multicolumn{2}{c||}{0.9258}
    & \multicolumn{2}{c|}{0.9844}
    \\ \hline
    \textbf{approximation ratio (decomposed, NF)}
    & \multicolumn{2}{c||}{0.9481}
    & \multicolumn{2}{c|}{0.9569}
    \\ \hline
    \textbf{approximation ratio (random guess)}
    & \multicolumn{4}{c|}{0.5038}
    \\ \hline \hline
    \textbf{device}
    & \multicolumn{4}{c|}{\texttt{Helios}}
    \\ \hline
    \textbf{RUS limit}
    & N/A & 1
    & N/A & 1
    \\ \hline
    \textbf{QEC cycles}
    & N/A & 0
    & N/A & 0
    \\ \hline
    \textbf{physical qubits}
    & 12 & 97
    & 12 & 97
    \\ \hline
    \textbf{physical 2-qubit gates}
    & 132 & 1132
    & 264 & 2132
    \\ \hline
    \textbf{estimated NF runtime (seconds)}
    & 0.267 & 1.64
    & 0.436 & 2.66
    \\ \hline \hline
    \textbf{success probability}
    & 0.136(5) & 0.131(17)
    & 0.180(7) & 0.034(10)
    \\ \hline
    \textbf{approximation ratio}
    & 0.9224(45) & 0.8743(64)
    & 0.9093(60) & 0.6991(115)
    \\ \hline
    \textbf{kept/total shots}
    & 500/500 & 411/500
    & 500/500 & 349/500
    \\ \hline
    \textbf{post-selection rate}
    & 1 & 0.822
    & 1 & 0.698
    \\ \hline
  \end{tabular}
  \label{tab:qaoa_po}
\end{table}

\begin{sidewaystable}[h]
  \centering
  \caption{%
    Results from running the decomposed HHL circuit in \cref{fig:hhl_compiled} on the \texttt{Helios} computer.
    Data with an effective RUS limit of $r_{\mathrm{eff}}$ is post-selected on the shots in which all RUS routines succeeded in at most $r_{\mathrm{eff}}$ attempts, thereby controlling for the effects of a different RUS limit in the program submitted to hardware.
    Uncertainties in the final digits of fidelities reflect 68\% confidence intervals determined by bootstrapping.
    The total number of shots for all experiments is 1000.
    In addition to the number of kept shots (those in which all state preparation subroutines succeeded, indicated by zero measurement outcomes on the appropriate flag/ancilla qubits, $\texttt{SP}=0$), more granular data is provided on the number of shots that were kept and have logical \texttt{QPE} and \texttt{EV} measurement outcomes as indicated.
    Fidelities are computed from the shots with $\texttt{SP} = \texttt{QPE} = \texttt{EV} = 0$.
  }
  \vspace{.5em}
  \setlength{\tabcolsep}{2pt}
  \renewcommand{\arraystretch}{1.2}
  \begin{tabular}{|l||c||c|c|c||c|c||c|c|c||c|c|c|}
    \hline
    {\textbf{logical qubits}}
    & \multicolumn{12}{c|}{$N = 5$}
    \\ \hline
    {\textbf{circuit parameters}}
    & \multicolumn{12}{c|}{$n_T = 14$}
    \\ \hline \hline
    {\textbf{device}}
    & \multicolumn{12}{c|}{\texttt{Helios}}
    \\ \hline
    {\textbf{RUS limit}}
    & N/A
    & \multicolumn{3}{c||}{1}
    & \multicolumn{2}{c||}{2}
    & \multicolumn{6}{c|}{3}
    \\ \hline
    \textbf{effective RUS limit}
    & N/A
    & \multicolumn{3}{c||}{1}
    & 1 & 2
    & 1 & 2 & 3
    & 1 & 2 & 3
    \\ \hline
    {\textbf{QEC cycles}}
    & N/A
    & 0 & 2 & 4
    & \multicolumn{2}{c||}{0}
    & \multicolumn{3}{c||}{0}
    & \multicolumn{3}{c|}{2}
    \\ \hline
    \textbf{physical qubits}
    & 5
    & \multicolumn{3}{c||}{48}
    & \multicolumn{2}{c||}{48}
    & \multicolumn{3}{c||}{48}
    & \multicolumn{3}{c|}{48}
    \\ \hline
    \textbf{physical 2-qubit gates}
    & 15
    & 677 & 749 & 821
    & 677 & $\le 1166$
    & 677 & $\le 1166$ & $\le 1655$
    & 749 & $\le 1282$ & $\le 1815$
    \\ \hline
    \textbf{estimated NF runtime (seconds)}
    & 0.09
    & 1.79 & 2.02 & 2.26
    & 1.83 & --
    & 1.83 & -- & --
    & 2.06 & -- & --
    \\ \hline \hline
    \textbf{fidelity $\mathcal{F}_{\ket{\chi}}$}
    & 0.9995(115)
    & 0.901(30) & 0.966(23) & 0.871(38)
    & 0.825(36) & 0.832(26)
    & 0.778(40) & 0.778(28) & 0.773(27)
    & 0.839(38) & 0.830(28) & 0.827(27)
    \\ \hline
    \textbf{kept shots ($\texttt{SP}=0$)}
    & 1000
    & 500 & 511 & 412
    & 505 & 970
    & 470 & 960 & 998
    & 488 & 967 & 999
    \\ \hline
    \textbf{kept shots w/} $\texttt{QPE}=0$, $\texttt{EV}=0$
    & 748
    & 272 & 293 & 223
    & 240 & 441
    & 226 & 449 & 466
    & 214 & 401 & 415
    \\ \hline
    \textbf{kept shots w/} $\texttt{QPE}=0$, $\texttt{EV}=1$
    & 238
    & 88 & 110 & 89
    & 104 & 209
    & 89 & 198 & 205
    & 89 & 188 & 194
    \\ \hline
    \textbf{kept shots w/} $\texttt{QPE}\ne0$, $\texttt{EV}=0$
    & 5
    & 93 & 65 & 64
    & 93 & 185
    & 93 & 174 & 184
    & 99 & 213 & 222
    \\ \hline
    \textbf{kept shots w/} $\texttt{QPE}\ne0$, $\texttt{EV}=1$
    & 9
    & 47 & 43 & 36
    & 68 & 135
    & 62 & 139 & 143
    & 86 & 165 & 168
    \\ \hline
  \end{tabular}
  \label{tab:hhl_data}
\end{sidewaystable}

\clearpage
\section*{Disclaimer}

This paper was prepared for informational purposes with contributions from the Global Technology Applied Research center of JPMorgan Chase \& Co. This paper is not a product of the Research Department of JPMorgan Chase \& Co. or its affiliates. Neither JPMorgan Chase \& Co. nor any of its affiliates makes any explicit or implied representation or warranty and none of them accept any liability in connection with this paper, including, without limitation, with respect to the completeness, accuracy, or reliability of the information contained herein and the potential legal, compliance, tax, or accounting effects thereof. This document is not intended as investment research or investment advice, or as a recommendation, offer, or solicitation for the purchase or sale of any security, financial instrument, financial product or service, or to be used in any way for evaluating the merits of participating in any transaction.

\end{document}